\documentclass[12pt,english]{article}
\usepackage{ae,aecompl}

\usepackage[T1]{fontenc}
\usepackage[latin9]{inputenc}
\usepackage{geometry}
\usepackage{array}
\usepackage{multirow}
\usepackage{amsthm}
\usepackage{amsmath}
\usepackage{amssymb}
\usepackage{setspace}
\usepackage{esint}
\usepackage[authoryear]{natbib}
\makeatletter

\providecommand{\tabularnewline}{\\}

\numberwithin{equation}{section}
\newcommand{\lyxaddress}[1]{
\par {\raggedright #1
\vspace{1.4em}
\noindent\par}
}
  \theoremstyle{plain}
  \newtheorem{assumption}{\protect\assumptionname}
\theoremstyle{plain}
\newtheorem{thm}{\protect\theoremname}
  \theoremstyle{plain}
  \newtheorem{cor}{\protect\corollaryname}
  \theoremstyle{remark}
  \newtheorem*{rem*}{\protect\remarkname}
  \theoremstyle{plain}
  \newtheorem{lem}{\protect\lemmaname}

\renewcommand\[{\begin{equation}}
\renewcommand\]{\end{equation}}

\makeatother

\usepackage{babel}
  \providecommand{\assumptionname}{Assumption}
  \providecommand{\lemmaname}{Lemma}
  \providecommand{\remarkname}{Remark}
\providecommand{\corollaryname}{Corollary}
\providecommand{\theoremname}{Theorem}

\begin{document}

\title{Adaptive Test of Conditional Moment Inequalities}

\author{By Denis Chetverikov%
\thanks{MIT, Economics Department. Email: dchetver@mit.edu. I thank Victor
Chernozhukov for his guidance, numerous discussions and permanent
support. I am also grateful to Isaiah Andrews, Jerry Hausman, Kengo
Kato, Anton Kolotilin, and Anna Mikusheva for useful comments and
discussions. The first version of the paper was presented at the Econometric
lunch at MIT in November, 2010.%
}}
\maketitle
\begin{abstract}
In this paper, I construct a new test of conditional moment inequalities,
which is based on studentized kernel estimates of moment functions
with many different values of the bandwidth parameter. The test automatically
adapts to the unknown smoothness of moment functions and has uniformly
correct asymptotic size. The test has high power in a large class
of models with conditional moment inequalities. Some existing tests
have nontrivial power against $n^{-1/2}$-local alternatives in a
certain class of these models whereas my method only allows for nontrivial
testing against $(n/\log n)^{-1/2}$-local alternatives in this class.
There exist, however, other classes of models with conditional moment
inequalities where the mentioned tests have much lower power in comparison
with the test developed in this paper.
\end{abstract}

\lyxaddress{Keywords: Conditional Moment Inequalities, Minimax Rate Optimality.}

\section{Introduction}

Conditional moment inequalities (CMI) are often encountered both in
economics and econometrics. In economics, they arise naturally in
many models that include behavioral choice, see \citet{Pakes2010}
for a survey. In these models, an agent chooses the action that maximizes
expected utility given her information set. Comparing the realized
action with any other available action leads to CMI. In econometrics,
they appear in the estimation problems with interval data and problems
with censoring, e.g., see \citet{ManskiandTamer2002}. In addition,
CMI offer a convenient way to study treatment effects in randomized
experiments as described in \citet{LeeandSongandWhang2011}. In the
next section, I provide three detailed examples of models with CMI.

Let $m:\,\mathbb{R}^{d}\times\mathbb{R}^{k}\times\Theta\rightarrow\mathbb{R}^{p}$
be a vector-valued known function. Let $(X,W)$ be a pair of $\mathbb{R}^{d}$
and $\mathbb{R}^{k}$-valued random vectors, and $\theta\in\Theta$
a parameter. The CMI can be written as
\begin{equation}
E[m(X,W,\theta)|X]\leq0\, a.s.\label{eq:former null hypothesis-1}
\end{equation}
where inequalities are understood piecewise. I am interested in testing
the null hypothesis, $H_{0}$, that $\theta=\theta_{0}$ against the
alternative, $H_{a}$, that $\theta\neq\theta_{0}$ based on iid sample
$(X_{i},W_{i})_{i=1}^{n}$ from the distribution of $(X,W)$. Note
that I also allow for conditional moment equalities since they can
be written as pairs of the CMI in model (\ref{eq:former null hypothesis-1}).

Using CMI for inference is difficult because often these inequalities
do not identify the parameter. Let 
\[
\Theta_{I}=\{\theta\in\Theta:\, E[m(X,W,\theta)|X]\leq0\,\, a.s.\}
\]
denote the identified set. The model is said to be identified if and
only if $\Theta_{I}$ is a singleton. Otherwise, CMI do not identify
the parameter $\theta$. For example, the latter may happen when the
CMI arise from a game-theoretic model with multiple equilibria. Moreover,
the parameter may be weakly identified. My approach leads to a test
with the correct asymptotic size no matter whether the parameter is
identified, weakly identified, or not identified.

Two approaches to robust CMI testing have been developed in the literature.
One approach (\citet{AndrewsandShi2010}), is based on converting
CMI into an infinite number of unconditional moment inequalities using
nonnegative weighting functions. The other approach (\citet{ChernozhukovLeeRosen2009}),
is based on estimating moment functions nonparametrically. My method
is inspired by the work of \citet{AndrewsandShi2010}. To motivate
the test developed in this paper, consider two examples of CMI models.
These models are highly stylized but convey main ideas. In the first
model, $m$ is multiplicatively separable in $\theta$, i.e. $m(X,W,\theta)=\theta\tilde{m}(X,W)$
for some $\tilde{m}:\,\mathbb{R}^{d}\times\mathbb{R}^{k}\rightarrow\mathbb{R}$
and $\theta\in\mathbb{R}$ with $E[\tilde{m}(X,W)|X]>0$ almost surely.
In the second model, $m$ is additively separable in $\theta$, i.e.
$m(X,W,\theta)=\tilde{m}(X,W)+\theta$. The identified sets, $\Theta_{I}$,
in these models are $\{\theta\in\mathbb{R}:\,\theta\leq0\}$ and $\{\theta\in\mathbb{R}:\,\theta\leq-\text{ess}\sup_{X}E[\tilde{m}(X,W)|X]\}$
correspondingly. \citet{AndrewsandShi2010} developed a test that
has nontrivial power against alternatives of the form $\theta_{0}=\theta_{0,n}=C/\sqrt{n}$
for any $C>0$ in the first model, so their test has extremely high
power in this model. It follows from \citet{Armstrong1} that their
test has low power in the second model, however (e.g., in comparison
with the test of \citet{ChernozhukovLeeRosen2009})%
\footnote{\citet{AndrewsandShi2010} developed tests based on both Cramer-von
Mises and Kolmogorov-Smirnov test statistics. In this paper, I mainly
refer to their test with Kolmogorov-Smirnov test statistic. Most statements
are also applicable for Cramer-von Mises test statistic as well, however.%
}. In constrast, I construct a test that has high power in a large
class of CMI models including models like that in the second example.
At the same time, my test has virtually the same power in models like
that described in the first example. The main difference between two
approaches is that my test statistic is based on the \textit{studentized}
estimates of moments whereas theirs is not. More precisely, \citet{AndrewsandShi2010}
also consider studentization but they modify the variance term so
that asymptotic power properties of their test are similar to those
of the test with no studentization. 

The test of \citet{ChernozhukovLeeRosen2009} also has high power
in a large class of CMI models but it requires knowledge of certain
smoothness properties of moment functions such as order of differentiability
whereas the test developed in this paper does not. Moreover, my test
automatically adapts to these smoothness properties selecting the
most appropriate weighting function. This feature of the test is important
because smoothness properties of moment functions are rarely known
in practice. For this reason, I call the test adaptive.

The test statistic in this paper is based on kernel estimates of moment
functions $E[m_{j}(X,W,\theta_{0})|X]$ with many bandwidth values
using positive kernels%
\footnote{A kernel is said to be positive if the kernel function is positive
on its support.%
}. Here $m_{j}(X,W,\theta)$ denotes $j$-th component of $m(X,W,\theta)$.
I assume that the set of bandwidth values expands as the sample size
$n$ increases so that the minimal bandwidth value converges to zero
at an appropriate rate while the maximal one is fixed. Since the variance
of the kernel estimators varies greatly with the bandwidth value,
each estimator is studentized, i.e. it is divided by its estimated
standard deviation. The test statistic, $\hat{T}$, is formed as the
maximum of these studentized estimates, and large values of $\hat{T}$
suggest that the null hypothesis is violated.

I develop a bootstrap method to simulate the critical value for the
test. The method is based on the observation that the distribution
of the test statistic, conditionally on the values $\{X_{i}\}_{i=1}^{n}$,
is asymptotically independent of the distribution of the noise $\{m(X_{i},W_{i},\theta_{0})-E[m(X_{i},W_{i},\theta_{0})|X_{i}]\}_{i=1}^{n}$
apart from its second moment. For reasons similar to those discussed
in \citet{Chernozhukov2007} and \citet{Andrews2010}, the distribution
of the test statistic in large samples depends heavily on the extent
to which CMI are binding. Moreover, the parameters that measure to
what extent CMI are binding can not be estimated consistently. I develop
a new approach to deal with this problem, which I refer to as the
refined moment selection (RMS) procedure. The approach is based on
the pretest that is used to decide what counterparts of the test statistic
should be used in simulating the critical value for the test. In comparison
with \citet{AndrewsandShi2010}, I use a model-specific critical value
for the pretest, which is simulated as a high quantile of the appropriate
distribution, whereas they use a deterministic threshold with no reference
to the model. For comparison reasons, I also provide a plug-in critical
value for the test. My proof of the bootstrap validity is interesting
on its own right because it is not known whether the test statistic
converges in distribution somewhere or not.

None of the tests in the literature including mine have power against
alternatives in the set $\Theta_{I}$. Therefore, I consider the alternatives
of the form
\begin{equation}
P\{E[m_{j}(X,W,\theta_{0})|X]>0\}>0\,\text{for some}\, j=1,...,p\label{eq: alternative}
\end{equation}
To show that my test has good power properties in a large class of
CMI models, I derive its power against alternatives of the form (\ref{eq: alternative})
assuming that $E[m(X,W,\theta_{0})|X]$ is some vector of unrestricted
nonparametric functions. In other words, I consider nonparametric
classes of alternatives. Once $m(X,W,\theta)$ is specified, it is
straightforward to translate my results into the parametric setting.
The test developed in this paper is consistent against any fixed alternative
outside of the set $\Theta_{I}$. I also show that my method allows
for nontrivial testing against $(n/\log n)^{-1/2}$-local one-directional
alternatives%
\footnote{In this paper, by one directional alternatives, I mean alternatives
of the form $E[m(X,W,\theta_{0})|X]=a_{n}f(X)$ for some sequence
of positive numbers $\{a_{n}\}_{n=1}^{\infty}$ converging to zero
where $f$ satisfies (\ref{eq: alternative}).%
}. Finally, I prove that the test is minimax rate optimal against certain
classes of smooth alternatives consisting of moment functions $E[m(X,W,\theta_{0})|X]$
that are sufficiently flat at the points of maxima. Minimax rate optimality
means that the test is uniformly consistent against alternatives in
the mentioned class whose distance from the set of models satisfying
(\ref{eq:former null hypothesis-1}) converges to zero at the fastest
possible rate. The requirement that functions should be sufficiently
flat can not be dropped because the test is based on the positive
kernels.

The literature concerned with unconditional and conditional moment
inequalities is expanding quickly. The list of published papers on
unconditional moment inequalities includes \citet{Chernozhukov2007},
\citet{RomanoShaikh2008}, \citet{Rosen2008}, \citet{AndrewsGuggenberger2009},
\citet{AndrewsHan2009}, \citet{Andrews2010}, \citet{Bugni2010},
\citet{Canay2010}, \citet{Pakes2010}, and \citet{RomanoShaikh2010}.
I note that there is also a large literature on partial identification
which is close related to that on moment inequalities. Methods specific
for conditional moment inequalities were developed in \citet{KhanAndTamer2009},
\citet{Kim2008}, \citet{ChernozhukovLeeRosen2009}, \citet{AndrewsandShi2010},
\citet{LeeandSongandWhang2011}, \citet{Armstrong1}, and \citet{Armstrong2}.
The case of CMI that point identify $\theta$ is treated in \citet{KhanAndTamer2009}.
The test of \citet{Kim2008} is closely related to that of \citet{AndrewsandShi2010}.
\citet{LeeandSongandWhang2011} developed a test based on the minimum
distance statistic in the one-sided $L_{p}$-norm and kernel estimates
of moment functions. The advantage of their approach comes from simplicity
of their critical value for the test, which is an appropriate quantile
of the standard Gaussian distribution. Their test is not adaptive,
however, since only one bandwidth value is used. \citet{Armstrong1}
developed a new method for computing the critical value for the test
statistic of \citet{AndrewsandShi2010} which leads to a more powerful
test than theirs but his method is not robust. In particular, his
method can not be used in the CMI models like that described in the
first example above. \citet{Armstrong2} considered the test statistic
similar to that used in this paper but he focused on estimation rather
than inference.

Finally, an important related paper in the statistical literature
is \citet{Dumbgen2001}. They consider testing qualitative hypotheses
in the ideal Gaussian white noise model where a researcher observes
a stochastic process that can be represented as a sum of the mean
function and a Brownian motion. In particular, they developed a test
for the null hypothesis that the mean function is (weakly) negative
almost everywhere. Even though their test statistic is somewhat related
to that used in this paper, the technical details of their analysis
are quite different.

The rest of the paper is organized as follows. The next section elaborates
on some examples of CMI models. Section \ref{sec:The-Test} formally
introduces the test. The main results of the paper are presented in
section \ref{sec:The-Main-Results}. A Monte Carlo simulation study
is described in section \ref{sec:Monte-Carlo-Results}. There I provide
an example of an alternative with the well-behaved moment function
such that the test developed in this paper rejects the null hypothesis
with probability higher than 80\% while the rejection probability
of all competing tests does not exceed 20\%. Brief conclusions are
drawn in section \ref{sec:Conclusions}. Finally, all proofs are contained
in the Appendix.

\section{\label{sec:Examples}Examples}

In this section, I provide three examples where CMI arise naturally
in economic and econometric models. The first two examples have function-valued
parameters. In order to fit these examples into my framework, one
can consider parametric approximations of corresponding functions.

\textbf{Incomplete Models of English Auctions. }My first example follows
\citet{HaileandTamer} treatment of English auctions under weak conditions.
The popular model of English auctions suggested by \citet{Milgrom-Weber1982}
assumes that each bidder is holding down the button while the price
is going up continuously until she wants to drop out. The price at
the moment of dropping out is her bid. In this model, it is well-known
that the dominant strategy is to make a bid equal to her valuation
of the object. In practice, participants usually call out bids, however.
So, the price rises in jumps, and the bid may not be equal to person's
valuation of the object. In this situation, the relation between bids
and valuations of the object depends crucially on the modeling assumptions.
\citet{HaileandTamer} derived certain bounds on the distribution
function of valuations based on minimal assumptions of rationality.

Suppose we have an auction with $m$ bidders whose valuations of the
object are drawn independently from the distrubution $F(\cdot,X)$
where $X$ denotes observable characterics of the object. Let $b_{1},...,b_{m}$
denote highest bids of each bidder. Let $b_{1:m}\leq...\leq b_{m:m}$
denote the ordered sequence of bids $b_{1},...,b_{m}$. Assuming that
bids do not exceed bidders' valuations, \citet{HaileandTamer} derived
the following upper bound on $F(\cdot,X)$:
\[
E[I\{b_{i:m}\leq v\}-\phi^{-1}(F(v,X))|X]\geq0\, a.s.
\]
for all $v\in\mathbb{R}$ and $i=1,...,m$ where $\phi(\cdot)$ is
a certain (known) function, see equation (3) in \citet{HaileandTamer}.
Similar lower bound follows from the assumption that bidders do not
allow oponents to win at a price they would like to beat. Assuming
we observe an iid sequence of auctions, these CMI can be used for
inference on $F(v,X)$.

\textbf{Interval Data. }In some cases, especially when data concerns
personal information like individual income or wealth, one has to
deal with interval data. Suppose we have a mean regression model
\[
Y=f(X,V)+\varepsilon
\]
where $E[\varepsilon|X,V]=0$ a.s. and $V$ is a scalar random variable.
Suppose that we observe $X$ and $Y$ but we do not observe $V$.
Instead, we observe $V_{0}$ and $V_{1}$ called brackets such that
$V\in(V_{0},V_{1})$ a.s. In empirical analysis, brackets may arise
because a respondent refuses to provide information on $V$ but provides
an interval to which $V$ belongs. Following \citet{ManskiandTamer2002}
assume that $f(X,V)$ is weakly increasing in $V$ and $E[Y|X,V]=E[Y|X,V,V_{0},V_{1}]$.
Then it is easy to see that
\[
E[I\{V_{1}\leq v\}(Y-f(X,v))|X,V_{0},V_{1}]\leq0
\]
and
\[
E[I\{V_{0}\geq v\}(Y-f(X,v))|X,V_{0},V_{1}]\geq0
\]
for all $v\in\mathbb{R}$. If we observe an iid sample from the model,
we can use these CMI for inference on $f(X,V)$.

\textbf{Treatment Effects. }Suppose we have a randomized experiment
where one group of people gets a new treatment while the control group
gets a placebo. Let $D=1$ if the person gets the treatment and $0$
otherwise. Let $p$ denote the probability that $D=1$. Let $X$ denote
person's observable characteristics and $Y$ denote a realized outcome.
Finally, let $Y_{0}$ and $Y_{1}$ denote counterfactual outcomes
had the person received a placebo or the new medicine respectively.
Then $Y=DY_{1}+(1-D)Y_{0}$. The question of interest is whether the
new medicine has a positive expected impact uniformly over all posible
person's charactersics $X$. In other words, the null hypothesis,
$H_{0}$, is that 
\begin{equation}
E[Y_{1}-Y_{0}|X]\geq0\, a.s.\label{eq: 11}
\end{equation}
Since in randomized experiments $D$ is independent of $X$, \citet{LeeandSongandWhang2011}
showed that
\begin{equation}
E[Y_{1}-Y_{0}|X]=E[DY/p-(1-D)Y/(1-p)|X]\label{eq: 12}
\end{equation}
Combining (\ref{eq: 11}) and (\ref{eq: 12}) gives CMI.

\section{\label{sec:The-Test}The Test}

In this section, I present the test statistic and give two bootstrap
methods to simulate a critical value. Given nonparametric nature of
the test, I use the corresponding terminology. For fixed $\theta_{0}$,
let $Y=m(X,W,\theta_{0})$, $f(X)=E[m(X,W,\theta_{0})|X]$, and $\varepsilon=Y-f(X)$
so that $E[\varepsilon|X]=0$ a.s. Then under the null hypothesis,
\[
f(X)\leq0\, a.s.
\]
I refer to $Y$ as a response variable, $f$ as a vector-valued regression
function, $X$ as a design point, and $\varepsilon$ as a disturbance.
Components of $f$ are denoted by $f_{1},...,f_{p}$.

The analysis in this paper is conducted conditionally on the set of
values $\{X_{i}\}_{i=1}^{n}$ of the insrumental variable $X$, so
all probabilistic statements in this paper should be understood conditionally
on $\{X_{i}\}_{i=1}^{n}$ for almost all sequences $\{X_{i}\}_{i=1}^{n}$.
Lemma \ref{lemma for support} in the Appendix provides certain conditions
that insure that assumptions used in this paper hold for almost all
sequences $\{X_{i}\}_{i=1}^{n}$.

Section \ref{sub:The-Test-Statistic} defines the test statistic assuming
that $E[\varepsilon_{i}\varepsilon_{i}^{T}]=\Sigma_{i}$ is known
for each $i=1,...,n$. Section \ref{sub:The-Test-Function} gives
two bootstrap methods to simulate a critical value. The first one
is based on plug-in asymptotics, and the second one is based on the
refined moment selection (RMS) procedure. Section \ref{sub:The-Test-Function}
also provides some intuition of why these procedures lead to the correct
asymptotic size of the test. When $\Sigma_{i}$ is not known, it should
be estimated from the data. Section \ref{sub:Estimating sigma} shows
how to construct an appropriate estimator $\hat{\Sigma}_{i}$ of $\Sigma_{i}$.
The feasible version of the test will be based on substituting $\hat{\Sigma}_{i}$
for $\Sigma_{i}$ both in the test statistic and in the critical value.
\ref{sub:Remarks} provides some notes on how to choose certain tuning
parameters.

\subsection{\label{sub:The-Test-Statistic}The Test Statistic}

The test statistic in this paper is based on the kernel estimator
of the vector-valued regression function $f$. Let $K:\,\mathbb{R}^{d}\rightarrow\mathbb{R}_{+}$
be some kernel. For bandwidth value $h\in\mathbb{R}_{+}$, denote
$K_{h}(x)=K(x/h)/h^{d}$. For each pair of observations $i,j=1,...,n$,
denote the weight function
\[
w_{h}(X_{i},X_{j})=\frac{K_{h}(X_{i}-X_{j})}{\sum_{k=1}^{n}K_{h}(X_{i}-X_{k})}
\]
Then the kernel estimator of $f_{m}(X_{i})$ is
\[
\hat{f}_{i,m,h}=\sum_{j=1}^{n}w_{h}(X_{i},X_{j})Y_{j,m}
\]
where $Y_{j,m}$ denotes $m$-th component of response variable $Y_{j}$.
Conditionally on $\{X_{i}\}_{i=1}^{n}$, the variance of the kernel
estimator $\hat{f}_{i,m,h}$ is
\[
V_{i,m,h}^{2}=\sum_{j=1}^{n}w_{h}^{2}(X_{i},X_{j})\Sigma_{j,mm}
\]
where $\Sigma_{j,m_{1}m_{2}}$ denotes $(m_{1},m_{2})$ component
of $\Sigma_{j}=E[\varepsilon_{j}\varepsilon_{j}^{T}]$.

Next, consider a finite set of bandwidth values $H=\{h=h_{\max}a^{k}:\, h\geq h_{\min},k=0,1,2,...\}$
for some $h_{\max}>h_{\min}$ and $a\in(0,1)$. For simplicity, I
assume that $h_{\min}=h_{\max}a^{k}$ for some $k\in\mathbb{N}$ so
that $h_{\min}$ is included in $H$. I assume that as the sample
size $n$ increases, $h_{\min}$ converges to zero while $h_{\max}$
is fixed. For each bandwidth value $h\in H$, choose a subset $I_{h}$
of observations such that $\Vert X_{i}-X_{j}\Vert>2h$ for all $i,j\in I_{h}$
with $i\neq j$ and for each $i=1,...,n$, there exist an element
$j(i)\in I_{h}$ such that $\Vert X_{i}-X_{j(i)}\Vert\leq2h$ where
$\Vert\cdot\Vert$ denotes the Eucledian norm on $\mathbb{R}^{d}$.
I refer to $I_{h}$ as a set of test points. The choice of $I_{h}$
may be random, but it is important to select $I_{h}$ independently
of response variables $\{Y_{i}\}_{i=1}^{n}$. So, conditionally on
$\{X_{i}\}_{i=1}^{n}$, I assume that $I_{h}$ is nonstochastic. It
will be assumed in the next section that $K(x)=0$ for any $x\in\mathbb{R}^{d}$
such that $\Vert x\Vert>1$. Thus, random variables $\{\hat{f}_{i,m,h}\}_{i\in I_{h}}$
are jointly independent for any fixed $m=1,...,p$ and $h\in H$ conditionally
on $\{X_{i}\}_{i=1}^{n}$. This fact will play a key role in the derivation
of the lower bound on the growth rate of the pdf of the test statistic,
which is used in the analysis of size properties of the test%
\footnote{Although my argument in the derivation of the lower bound is based
on the fact that $\{\hat{f}_{i,m,h}\}_{i\in I_{h}}$ are jointly independent,
I believe that the same lower bound can be obtained even for the case
$I_{h}=\{1,...,n\}$. If this statement is true, one can use $I_{h}=\{1,...,n\}$
in the definition of the test statistic.%
}. Finally, denote $S=\{(i,m,h):\, h\in H,i\in I_{h},m=1,...,p\}$.

Based on this notation, the test statistic is
\[
T=\max_{s\in S}\frac{\hat{f}_{s}}{V_{s}}
\]

Let me now explain why the optimal bandwidth value depends on the
smoothness properties of the components $f_{1},...,f_{p}$ of $f$.
Without loss of generality, consider $j=1$. Suppose that $f_{1}(X)$
is flat. Then $f_{1}(X)$ is positive on the large subset of its domain
whenever its maximal value is positive. Hence, the maximum of $\hat{T}$
will correspond to a large bandwidth value because the variance of
the kernel estimator, which enters the denominator of the test statistic,
decreases with the bandwidth value. On the other hand, if $f_{1}(X)$
is allowed to have peaks, then there may not exist a large subset
where it is positive. So, large bandwidth values may not yield large
values of $\hat{T}$, and small bandwidth values should be used in
such cases. I circumvent these problems by considering the set of
bandwidth values jointly, and let the data determine the best bandwidth
value. In this sense, my test adapts to the smoothness properties
of $f(X)$. This allows me to construct a test with good uniform power
properties over possible smoothness of $f(X)$.

When $\Sigma_{i}$ is not observed, which is usually the case in practice,
one can define $\hat{V}_{i,m,h}^{2}=\sum_{j=1}^{n}w_{h}^{2}(X_{i},X_{j})\hat{\Sigma}_{j,mm}$
and use

\[
\hat{T}=\max_{s\in S}\frac{\hat{f}_{s}}{\hat{V}_{s}}
\]
instead of $T$ where $\hat{\Sigma}_{j}$ is some estimator of $\Sigma_{j}$.
Some possible estimators are discussed in section \ref{sub:Estimating sigma}.

\subsection{\label{sub:The-Test-Function}Critical Values}

Suppose we want to construct a test of size $\alpha$. This subsection
explains how to simulate a critical value $t_{1-\alpha}$ for the
statistic $\hat{T}$ based on two bootstrap methods. One method is
based on the plug-in asymptotics, and the other one is based on the
refined moment selction (RMS) procedure. Both methods have deterministic
and randomized versions. For the randomized versions, one first determines
some small interval, say $[c,c+\beta]$ with $\beta>0$, where the
critical value belongs. Then one draws the critical value from a certain
distribution with the support $[c,c+\beta]$. This randomization comes
from my proof technique, which is based on the Linderberg method.
Under somewhat stronger conditions, I also prove the validity of both
methods with $\beta=0$, which corresponds to their deterministic
versions. The test will be of the following form: reject the null
hypothesis if and only if $\hat{T}>t_{1-\alpha}$.

Let $\beta$ be either zero or some small positive number. Let $g_{0}$
be a thrice differentiable function from $\mathbb{R}$ into $[0,1]$
such that $g_{0}(x)=1$ for all $x\leq0$ and $g_{0}(x)=0$ for all
$x\geq1$. Denote $g(x)=g_{0}((x-c)/\beta)$ for some $c\in\mathbb{R}$.
Since $g(x)\in[0,1]$ for all $x\in\mathbb{R}$, $g(\cdot)$ gives
a randomized test: upon observing the test statistic $\hat{T}=x$,
one accepts the null hypothesis with probability $g(x)$. I will choose
$c$ so that, under the null hypothesis, $E[g(\hat{T})]\geq1-\alpha+o(1)$
as $n\rightarrow\infty$, which leads to the correct asymptotic size
of this randomized test. An equivalent way to describe this test is
as follows. Let $U$ be a random variable independent of the data
with uniform distribution on $[0,1]$. Define the critical value $t_{1-\alpha}$
for the test from the equation $g(t_{1-\alpha})=U$. Since $g(x)$
is decreasing in $x$, this equation has the unique solution so that
$t_{1-\alpha}$ is well-defined. Lemma \ref{lemma: equivalence of tests}
in the Appendix shows that $E[g(\hat{T})]=P\{\hat{T}\leq t_{1-\alpha}\}$,
which means that the randomized test is equivalent to the test based
on the critical value $t_{1-\alpha}$. Note that the latter formulation
is more convenient for the confidence set construction: one can use
the same $U$ for all possible values of $\theta_{0}$. For the purposes
of presentation, the former formulation is suitable, however. I refer
to $g(\cdot)$ as a test function.

Let me now describe two possible bootstrap methods to simulate $c$.
The first method is based on plug-in asymptotics. It relies on two
observations. First, it is easy to see that, for a fixed distribution
of disturbances $\{\varepsilon_{i}\}_{i=1}^{n}$, the maximum of $1-\alpha$
quantile of the test statistic $\hat{T}$ over all possible functions
$f$ satisfying $f\leq0$ almost surely corresponds to $f=0_{p}$.
Second, lemmas \ref{variance} and \ref{lemma: aplication of Chatterjee}
in the Appendix show that the distribution of the statistic $\hat{T}$
is asymptotically independent of the distrubution of disturbances
$\{\varepsilon_{i}:\, i=1,...,n\}$ apart from their second moments
$\{\Sigma_{i}:\, i=1,...,n\}$. These observations suggest that one
can simulate $c$ by the following procedure:
\begin{enumerate}
\item For each $i=1,...,n$, simulate $\tilde{Y}_{i}\sim N(0_{p},\hat{\Sigma}_{i})$
independently across $i$.
\item Calculate $T^{PIA}=\max_{(i,m,h)\in S}\sum_{j=1}^{n}w_{h}(X_{i},X_{j})\tilde{Y}_{j,m}/\hat{V}_{i,m,h}$.
\item Repeat steps 1 and 2 independently $B$ times for some large $B$
to obtain $\{T_{b}^{PIA}:\, b=1,...,B\}$.
\item Find $c_{1-\alpha}^{PIA}$ such that $\sum_{b=1}^{B}g_{0}((T_{b}^{PIA}-c_{1-\alpha}^{PIA})/\beta)/B=1-\alpha$. 
\end{enumerate}
Then plug-in test function $g_{1-\alpha}^{PIA}:\,\mathbb{R}\rightarrow[0,1]$
is given by $g_{1-\alpha}^{PIA}(x)=g_{0}((x-c_{1-\alpha}^{PIA})/\beta)$
for all $x\in\mathbb{R}$.

The second method is based on the refined moment selection (RMS) procedure.
It gives a less conservative critical value while maintaining the
required size of the test. The method is based on the observation
that $|\hat{T}|=O_{p}(\sqrt{\log n})$ if $f=0_{p}$ (see lemmas \ref{normal conv},
\ref{variance}, \ref{lemma: aplication of Chatterjee} in the Appendix)
while $\hat{f}_{i,m,h}/\hat{V}_{i,m,h}\rightarrow-\infty$ with a
polynomial rate if $f_{m}(X_{i})<0$ and $h\rightarrow0$. Such terms
will have asymptotically negligible effect on the distribution of
$\hat{T}$, so we can ignore corresponding terms in the simulated
statistic. Specifically, let $\gamma<\alpha/2$ be some small positive
number. First, use the plug-in bootstrap to find $c_{1-\gamma}^{PIA}$.
Denote
\[
S^{RMS}=\{s\in S:\,\hat{f}_{s}/\hat{V}_{s}>-2(c_{1-\gamma}^{PIA}+\beta)\}
\]
Second, run the following procedure:
\begin{enumerate}
\item For each $i=1,...,n$, simulate $\tilde{Y}_{i}\sim N(0_{p},\hat{\Sigma}_{i})$
independently across $i$.
\item Calculate $T^{RMS}=\max_{(i,m,h)\in S^{RMS}}\sum_{j=1}^{n}w_{h}(X_{i},X_{j})\tilde{Y}_{j,m}/\hat{V}_{i,m,h}$.
\item Repeat steps 1 and 2 independently $B$ times for some large $B$
to obtain $\{T_{b}^{RMS}:\, b=1,...,B\}$.
\item Find $c_{1-\alpha+2\gamma}^{RMS}$ such that $\sum_{b=1}^{B}g_{0}((T_{b}^{RMS}-c_{1-\alpha+2\gamma}^{RMS})/\beta)/B=1-\alpha+2\gamma$.
\end{enumerate}
Then RMS test function $g_{1-\alpha}^{RMS}:\,\mathbb{R}\rightarrow\mathbb{R}$
is given by $g_{1-\alpha}^{RMS}(x)=g_{0}((x-c_{1-\alpha+2\gamma}^{RMS})/\beta$
for all $x\in\mathbb{R}$. The additional term $2\gamma$ can be interpreted
as a correction for the truncation procedure introduced in $S^{RMS}$.

\subsection{\label{sub:Estimating sigma}Estimating $\Sigma_{i}$}

Let me now explain how one can estimate $\Sigma_{i}$. The literature
on estimating $\Sigma_{i}$ is huge. Among other papers, it includes
\citet{Rice1984}, \citet{Muller1987}, \citet{HardleTsybakov2007},
and \citet{FanYao1998}. For scalar-valued response variables, a variaty
of such estimators is described in \citet{Horowitz2001}. All those
estimators can be immediately generalized to vector-valued response
variables. For completeness, I describe one estimator here. For $i=1,...,n$
define $j(i)$ by the following recursion:

\[
j(1)=\arg\min_{j=2,...,n}\Vert X_{j}-X_{1}\Vert
\]
and

\[
j(i)=\arg\min_{j\neq i,\, j(1),...,j(i-1)}\Vert X_{j}-X_{i}\Vert
\]
Then variance $\Sigma_{i}$ can be estimated by

\[
\hat{\Sigma}_{i}=\frac{\sum_{k=1}^{n}(Y_{k}-Y_{j(k)})(Y_{k}-Y_{j(k)})^{T}I(\Vert X_{k}-X_{i}\Vert\leq b_{n})}{2\sum_{k=1}^{n}I(\Vert X_{k}-X_{i}\Vert\leq b_{n})}
\]
where $b_{n}$ denotes some bandwidth value. This estimator will be
uniformly consistent for $\Sigma_{i}$ over $i=1,...,n$ with rate
$(\log n/n)^{1/(2+d)}$, i.e.
\[
\max_{i=1,...,n}\Vert\hat{\Sigma}_{i}-\Sigma_{i}\Vert_{o}=O_{p}\left(\frac{\log n}{n}\right)^{1/(2+d)}
\]
if (i) $b_{n}\asymp(\log n/n)^{1/(2+d)}$ and (ii) assumptions from
section \ref{sub:Assumptions} hold where $\Vert\cdot\Vert_{o}$ denotes
the spectral norm on the space of $p\times p$-dimensional symmetric
matrices corresponding to Eucledian norm on $\mathbb{R}^{p}$. To
choose bandwidth value $b_{n}$ in practice, one can use any type
of the cross validation. An advantage of this estimator is that it
is fully adaptive with respect to smoothness properties of regression
function $f$.

The intuition behind this estimator is based on the following argument.
Note that $j(k)$ is chosen so that $X_{j(k)}$ is close to $X_{k}$.
If regression function $f$ is continuous, 
\[
Y_{k}-Y_{j(k)}=f(X_{k})-f(X_{j(k)})+\varepsilon_{k}-\varepsilon_{j(k)}\approx\varepsilon_{k}-\varepsilon_{j(k)}
\]
so that
\[
E[(Y_{k}-Y_{j(k)})(Y_{k}-Y_{j(k)})^{T}]\approx\Sigma_{k}+\Sigma_{j(k)}
\]
since $\varepsilon_{k}$ is independent of $\varepsilon_{j(k)}$.
If $b_{n}$ is small enough and $\Sigma(X)$ is continuous, $\Sigma_{k}+\Sigma_{j(k)}\approx2\Sigma_{i}$
since only $X_{k}$ satisfying $\Vert X_{k}-X_{i}\Vert\leq b_{n}$
are used in estimating $\Sigma_{i}$.

\subsection{\label{sub:Remarks}Remarks on the Choice of Testing Parameters}

Implementing the deterministic version of the test requires choosing
minimal and maximal bandwidth values $h_{\min}$ and $h_{\max}$ and
the parameter $\gamma$. The randomized version of the test also use
the parameter $\beta$ and the function $g_{0}:\,\mathbb{R}\rightarrow[0,1]$.
In this section, I provide some notes on how to choose these objects
for the randomized test to make sure that the test maintains the required
size.

First, I recommend to set $h_{\max}=\max_{i,j=1,...,n}\Vert X_{i}-X_{j}\Vert/2$
as a normalization. Second, it follows from theorem \ref{thm: size}
that the test with RMS test function is not conservative asymptotically
only if $\gamma=\gamma_{n}\rightarrow0$ as $n\rightarrow0$. So,
I recommend to set $\gamma$ as a small fraction of $\alpha$, for
example $\gamma=0.01$ for $\alpha=0.05$. Alternatively, one can
set $\gamma=0.1/\log(n)$ similarly the corresponding choice in \citet{ChernozhukovLeeRosen2009}.

Next, consider how to choose $g_{0}$, $h_{\min}$, and $\beta$.
It follows from theorems \ref{thm: size} and \ref{thm:Chatterjee}
and lemma \ref{lemma: aplication of Chatterjee} that the test maintains
the required size if
\[
\Delta=\frac{3}{6^{1/3}\beta^{2/3}}pbn^{1/3}\left(\frac{\Vert g_{0}^{\prime\prime\prime}\Vert_{\infty}}{\beta^{3}}+\frac{3\Vert g_{0}^{\prime\prime}\Vert_{\infty}}{\beta^{2}}+\frac{\Vert g_{0}^{\prime}\Vert_{\infty}}{\beta}\right)^{1/3}(\Vert g_{0}^{\prime}\Vert_{\infty}\log|S|)^{2/3}F
\]
is small in comparison with $\alpha$ (required size) where 
\[
F=\left(\max E[|\varepsilon_{i,m}^{3}|]+\max\sqrt{8/\pi}\Sigma_{i,mm}^{3/2}\right)^{1/3}
\]
with both maxima taken over $i=1,...,n$ and $m=1,...,p$ and
\[
b=\max_{(i,m,h)\in S;\, j=1,...,n}\frac{w_{h}(X_{i},X_{j})}{V_{i,m,h}}
\]
If $\beta\ll1$, the good choice of $g_{0}$ is given by
\[
g_{0}(x)=\begin{cases}
1 & \text{if}\, x\leq0\\
1-(16/3)x^{3} & \text{if}\, x\in(0,1/4]\\
7/6-x-4(x-1/4)^{2}+(16/3)(x-1/4)^{3} & \text{if}\, x\in(1/4,3/4]\\
(16/3)(1-x)^{3} & \text{if}\, x\in(3/4,1]\\
0 & \text{if}\, x>1
\end{cases}
\]
This function is chosen so that $g_{0}^{\prime\prime\prime}(x)=-32$
for $x\in(0,1/4]$, $+32$ for $x\in(1/4,3/4]$, and $-32$ for $x\in(3/4,1]$.
Given this function, if $\beta\leq1$, it is enough to set parameters
so that 
\begin{equation}
1.8pbn^{1/3}(\log|S|)^{2/3}F/\beta^{5/3}\ll\alpha\label{eq:guarantee}
\end{equation}
Given $h_{\min}$, $b$ and $F$ can be estimated from the data. Then
one can choose $\beta$ so that the inequality above is satisfied.
Note that there is a trade-off between choosing small $\beta$ and
small $h_{\min}$ since $b$ is a decreasing function of $h_{\min}$. 

I note that the inequality (\ref{eq:guarantee}) guarantees good size
properties of the test uniformly over a large set of the true distributions
of disturbances $\{\varepsilon_{j}\}_{j=1}^{n}$. In particular, this
set includes discrete distributions, which lead to the distributions
of the test statistic that are difficult to approximate using Gaussian
disturbances%
\footnote{Similar phenomenon is also known in the classical theory of Central
Limit Theorems, see \citet{Ibragimov-Linnik}%
}. Therefore, this inequality is difficult to satisfy in sample sizes
typical for economic data. Nevertheless, this inequality is still
useful because it gives a starting point in choosing testing parameters.

\section{\label{sec:The-Main-Results}The Main Results}

This section presents my main results. Section \ref{sub:Assumptions}
gives regularity conditions. Section \ref{sub:Size} describes size
properties of the test. Section \ref{sub:Consistency-Fixed} explains
the behavior of the test under a fixed alternative. Section \ref{sub:Consistency-One-Directional}
derives the rate of consistency of the test against one-directional
alternatives mentioned in the introduction. Section \ref{sub:Uniform-Consistency}
shows the rate of uniform consistency against certain classes of smooth
alternatives. Section \ref{sub:Lower-Bound} presents the minimax
rate-optimality result.

\subsection{\label{sub:Assumptions}Assumptions}

Let $M_{h}(X_{i})$ be the number of elements in the set $\{X_{j}:\,\Vert X_{j}-X_{i}\Vert\leq h,\, j=1,...,n\}$.
In what follows, I will write $C$ and its variants for a generic
constant whose value may vary depending on the context. Results in
this paper will be proven under the following regularity assumptions.
\begin{assumption}
\label{ass:Design-points}(i) Design points $\{X_{i}\}_{i=1}^{n}$
are nonstochastic. (ii) For some constant $0<\bar{C}<\infty$ and
all $i=1,...,n$, $\Vert X_{i}\Vert<\bar{C}$. (iii) For some constants
$0<C_{1}<C_{2}<\infty$, $C_{1}nh^{d}\leq M_{h}(X_{i})\leq C_{2}nh^{d}$
for all $i\in\mathbb{N}$ and $h\in H=H_{n}$.
\end{assumption}
The design points are nonstochastic because the analysis is conducted
conditionally on $\{X_{i}\}_{i=1}^{n}$. Assumption 1 also states
that the design points have bounded support, which is a mild assumption.
In addition, it states that the number of design points in certain
neighborhoods of each design point is proportional to the volume of
the neighborhood with the coefficient of proportionality bounded from
above and away from zero. It is stated in \citet{Horowitz2001} that
assumption \ref{ass:Design-points} holds in an iid setting with probability
approaching one as the sample size increases if the distribution of
$X_{i}$ is absolutely continuous with respect to Lebegue measure,
has bounded support, and has the density bounded away from zero on
the support. This statement is actually wrong unless one makes some
extra assumptions. Lemma \ref{lem:wrong statement} in the Appendix
gives a counter-example. Instead, lemma \ref{lemma for support} shows
that assumption \ref{ass:Design-points} holds for large $n$ almost
surely if, in addition, I assume that the density of $X_{i}$ is bounded
from above, and that the support of $X_{i}$ is a convex set. Necessity
of the density boundedness is obvious. Convexity of the support is
not necessary for assumption \ref{ass:Design-points} but it gives
a good trade-off between generality and simplicity. In general, one
should deal with some smoothness properties of the boundary of the
support. Note that the statement ``for large $n$ almost surely''
is stronger than ``with probability approaching one''. Note also
that assumption \ref{ass:Design-points}(iii) requires inequalities
to hold for all $i\in\mathbb{N}$, not just for $i=1,...,n$.
\begin{assumption}
\label{ass:Disturbances}(i) Disturbances $\{\varepsilon_{i}:\, i=1,...,n\}$
are independent $\mathbb{R}^{p}$-valued random variables with $\mathbb{E}[\varepsilon_{i,m_{1}}]=0$,
$\mathbb{E}[\varepsilon_{i,m_{1}}\varepsilon_{i,m_{2}}]=\Sigma_{i,m_{1}m_{2}}<\infty$,
and $\mathbb{E}[\varepsilon_{i,m_{1}}\varepsilon_{i,m_{2}}\varepsilon_{i,m_{3}}\varepsilon_{i,m_{4}}]=s_{i,m_{1}m_{2}m_{3}m_{4}}^{4}<\infty$
for all $i=1,...,n$ and $m_{1},m_{2},m_{3},m_{4}=1,...,p$. (ii)
For some constants $0<C<\infty$ and $\delta>0$, $\mathbb{E}[|\varepsilon_{i,m}|^{4+\delta}]\leq C$
for all $i=1,...,n$ and $m=1,...,p$. (iii) For some constant $0<C<\infty$,
$|\Sigma_{i,m_{1}m_{2}}-\Sigma_{j,m_{1}m_{2}}|\leq C\Vert X_{i}-X_{j}\Vert$
and $|s_{i,m_{1}m_{2}m_{3}m_{4}}^{4}-s_{j,m_{1}m_{2}m_{3}m_{4}}^{4}|\leq C\Vert X_{i}-X_{j}\Vert$
for all $i,j=1,...,n$ and $m_{1},m_{2},m_{3},m_{4}=1,...,p$. (iv)
For some constant $0<C<\infty$, $\Sigma_{i,mm}\geq C$ for all $i=1,...,n$
and $m=1,...,p$.
\end{assumption}
The reason for imposing assumption \ref{ass:Disturbances} is threefold.
First, finite third moment of disturbances is used in the derivation
of a certain invariance principle with the rate of convergence. As
in the classical central limit theorem, finite two moments are sufficient
to prove weak convergence but more finite moments are necessary if
we are interested in the rate of convergence. Second, finite $4+\delta$
moment of disturbances and Lipshitz continuity properties are used
to make sure that $\hat{\Sigma}_{i}$ converges in probability to
$\Sigma_{i}$ uniformly over $i=1,...,n$ for a particular estimator
$\hat{\Sigma}_{i}$ of $\Sigma_{i}$ described in section \ref{sub:Estimating sigma}
at an appropriate rate. Finally, I assume that the variance of each
component of disturbances is bounded away from zero for simplicity
of the presentation. Since I use a studentization of kernel estimators,
without this assumption, it would be necessary to truncate the variance
of the kernel estimators from below with truncation level slowly converging
to zero. That would complicate the derivation of the main results
without changing main ideas.

Before stating assumption \ref{ass:Regression-function}, let me give
formal definitions of Holder smoothness class $\mathcal{F}(\tau,L)$
and its subsets $\mathcal{F}_{\varsigma}(\tau,L)$. For $d$-tuple
of nonnegative integers $\alpha=(\alpha_{1},...,\alpha_{d})$ with
$|\alpha|=\alpha_{1}+...+\alpha_{d}$, function $g:\,\mathbb{R}^{d}\rightarrow\mathbb{R}$,
and $x=(x_{1},...,x_{d})\in\mathbb{R}^{d}$, denote 
\[
D^{\alpha}g(x)=\frac{\partial^{|\alpha|}g}{\partial x_{1}^{\alpha_{1}}...\partial x_{d}^{\alpha_{d}}}(x)
\]
whenever it exists. For $\tau>0$, it is said that the function $g:\,\mathbb{R}^{d}\rightarrow\mathbb{R}$
belongs to the class $\mathcal{F}(\tau,L)$ if it has continuous partial
derivatives upto order $[\tau]$ and for any $\alpha=(\alpha_{1},...,\alpha_{d})$
such that $|\alpha|=[\tau]$ and $x,y\in\mathbb{R}^{d}$,
\[
|D^{\alpha}g(x)-D^{\alpha}g(y)|\leq\Vert x-y\Vert^{\tau-[\tau]}
\]
Here $[\tau]$ denotes the largest integer strictly smaller than $\tau$.
For any $g\in\mathcal{F}(\tau,L)$, $x=(x_{1},...,x_{d})\in\mathbb{R}^{d}$,
and $l=(l_{1},...,l_{d})\in\mathbb{R}^{d}$ satisfying $\sum_{m=1}^{d}l_{m}^{2}=1$,
let $g^{(k,l)}(x)$ denote $k$-th derivative of function $f$ in
direction $l$ at point $x$ whenever it exists. For $\varsigma=1,...,[\tau]$,
let $\mathcal{F}_{\varsigma}(\tau,L)$ denote the class of all elements
of $\mathcal{F}(\tau,L)$ such that for any $g\in\mathcal{F}_{\varsigma}(\tau,L)$
and $l=(l_{1},...,l_{d})\in\mathbb{R}^{d}$ satisfying $\sum_{m=1}^{d}l_{m}^{2}=1$,
$f^{(k,l)}(x)=0$ for all $k=1,...,\varsigma$ whenever $f^{(1,l)}(x)=0$,
and there exist $x=(x_{1},...,x_{d})\in\mathbb{R}^{d}$ and $l=(l_{1},...,l_{d})\in\mathbb{R}^{d}$
satisfying $\sum_{m=1}^{d}l_{m}^{2}=1$ such that $f^{(\varsigma+1,l)}(x)\neq0$
and $f^{(1,l)}(x)=0$. If $\tau\leq1$, I set $\varsigma=0$ and $\mathcal{F}_{\varsigma}(\tau,L)=\mathcal{F}(\tau,L)$.
\begin{assumption}
\label{ass:Regression-function}(i) For some $\tau\geq1/4$, $L>0$,
and $\varsigma=1,...,[\tau]$, regression functions $f_{m}(\cdot)=f_{m,n}(\cdot)$
belong to the class $\mathcal{F}_{\varsigma}(\tau,L)$ for all $m=1,...,p$.
(ii) If $\varsigma<[\tau]$, then for any $x\in\mathbb{R}^{d}$ and
all $\alpha=(\alpha_{1},...,\alpha_{d})$ such that $|\alpha|=\varsigma+1$,
$|D^{\alpha}f_{m}(x)|\leq C$ for some constant $C>0$ and all $m=1,...,p$.
\end{assumption}
For simplicity of notation, I assume that all components of $f$ have
the same smoothness properties. This assumption is used in the derivation
of the power properties of the test. The restriction $\tau\geq1/4$
is also needed to make sure that $\hat{\Sigma}_{i}$ converges in
probability to $\Sigma_{i}$ uniformly over $i=1,...,n$ at an appropriate
rate. I allow regression functions to depend on $n$ to perform a
local power analysis.
\begin{assumption}
\label{ass:bandwidth values}Set of bandwidth values has the following
form: $H=H_{n}=\{h=h_{\max}a^{k}:\, h\geq h_{\min},k=0,1,2,...\}$
where $a\in(0,1)$, $h_{\max}=\bar{C}$ and $h_{\min}=h_{\min,n}\rightarrow0$
as $n\rightarrow\infty$ such that $|H_{n}|\leq C\log n$ for some
constant $C>0$.
\end{assumption}
According to this assumption, maximal bandwidth value, $h_{\max}$,
is independent of $n$. Its value is chosen to match the radius $\bar{C}$
of the support of design points. It is intented to detect deviations
from the null hypothesis in the form of flat alternatives. Minimal
bandwidth value, $h_{\min}$, converges to zero as the sample size
increases in such a way that the number of bandwidth values in the
set $H_{n}$ is growing at a logarithmic rate or slower. This assumption
will be satisfied if $h_{\min}$ converges to zero at a polynomial
rate. Minimal bandwidth value is intended to detect deviations from
the null hypothesis in the form of alternatives with peaks.
\begin{assumption}
\label{ass:variance estimator}Estimators $\hat{\Sigma}_{i}$ of $\Sigma_{i}$
satisfy $\max_{i=1,...,n}\Vert\hat{\Sigma}_{i}-\Sigma_{i}\Vert_{o}=o_{p}(n^{-\kappa})$
with $\kappa=1/(2+d)-\phi$ for arbitrarily small $\phi>0$ where
$\Vert\cdot\Vert_{o}$ denotes the spectral norm on the space of $p\times p$-dimensional
symmetric matrices corresponding to the Euclidean norm on $\mathbb{R}^{p}$.
\end{assumption}
As follows from \citet{Muller1987}, under assumptions \ref{ass:Disturbances}
and \ref{ass:Regression-function}, assumption \ref{ass:variance estimator}
is satisfied for the estimators $\hat{\Sigma}_{i}$ of $\Sigma_{i}$
described in section \ref{sub:Estimating sigma}. In practice, due
to the course of dimensionality, it might be useful to use some parametric
or semi-parametric estimators of $\Sigma_{i}$ instead of the estimator
described in section \ref{sub:Estimating sigma}. For example, if
we assume that $\Sigma_{i}=\Sigma_{j}$ for all $i,j=1,...,n$, then
the estimator of \citet{Rice1984} (or its multivariate generalization)
is $1/\sqrt{n}$-consistent. In this case, assumption \ref{ass:variance estimator}
will be satisfied with $\kappa=1/2-\phi$ for arbitrarily small $\phi>0$.
\begin{assumption}
\label{ass:The-kernel}(i) The kernel $K$ is positive and supported
on $\{x\in\mathbb{R}^{d}:\,\Vert x\Vert\leq1\}$. (ii) For some constant
$0<C<1$, $K(x)\leq1$ for all $x\in\mathbb{R}^{d}$ and $K(x)\geq C$
for all $\Vert x\Vert\leq1/2$.
\end{assumption}
I assume that the kernel function is positive on its support. Many
kernels satisfy this assumption. For example, one can use rectangular,
triangular, parabolic, or biweight kernels. See \citet{Tsybakov_book}
for the definitions. On the other hand, the requirement that the kernel
is positive on its support excludes higher-order kernels, which are
necessary to achieve minimax optimal testing rate over large classes
of smooth alternatives. I require positive kernels because of their
negativity-invariance property, which means that any kernel smoother
with a positive kernel maps the space of negative functions into itself.
This property is essential for obtaining a test with the correct asymptotic
size when smoothness properties of moment functions are unknown. With
higher-order kernels, one has to assume undersmoothing so that the
bias of the estimator is asymptotically negligible in comparison with
its standard deviation. Otherwise, large values of $\hat{T}$ might
be caused by large values of the bias term relative to the standard
deviation of the estimator even though all components of $f(X)$ are
negative. However, for undersmoothing, one has to know the smoothness
properties of $f(X)$. In constrast, with positive kernels, the set
of bandwidth values can be chosen without reference to these smoothness
properties. In particular, the largest bandwidth value can be chosen
to be independent of the sample size $n$. Nevertheless, the test
developed in this paper will be rate optimal in the minimax sense
against class $\mathcal{F}_{[\tau]}(\tau,L)$ when $\tau>d$.
\begin{assumption}
\label{ass: test function}(i) For some constant $C>0$, $\beta=\beta_{n}\leq C$.
(ii) $(\log n)^{4}/(\beta_{n}^{10}h_{\min}^{3d}n)\rightarrow0$ as
$n\rightarrow\infty$.
\end{assumption}
Assumption \ref{ass: test function} establishes the trade-off between
choosing small value of $\beta$ and small value of $h_{\min}$. It
is a key condition used to establish an invariance principle that
shows that asymptotic distribution of $\hat{T}$ depends on the distribution
of disturbances $\{\varepsilon_{i}:\, i=1,...,n\}$ only through their
covariances $\{\Sigma_{i}:\, i=1,...,n\}$. Under somewhat stronger
conditions, corollary \ref{Corr: deterministic version} shows that
I can set $\beta=0$, which corresponds to the determinstic version
of the test. Note that from assumption \ref{ass: test function}(ii),
it follows that $h_{\min}$ converges to zero at a polynomial rate
which is consistent with assumption \ref{ass:bandwidth values}.
\begin{assumption}
\label{ass:Choice function}(i) For every $h\in H_{n}$, set of test
points $I_{h}=I_{h,n}$ is such that $\Vert X_{i}-X_{j}\Vert>2h$
for all $i,j\in I_{h,n}$ with $i\neq j$ and for each $i=1,...,n$,
there exists an element $j(i)\in I_{h,n}$ such that $\Vert X_{i}-X_{j(i)}\Vert\leq2h$.
(ii) $S=S_{n}=\{(i,m,h):\, h\in H_{n},i\in I_{h,n},m=1,...,p\}$.
\end{assumption}
Denote the class of models satisfying assumptions \ref{ass:Disturbances}
and \ref{ass:Regression-function} for some fixed values of all constants
by $\mathcal{G}$. Each element $w\in\mathcal{G}$ consists of a pair
$(f^{w},\varepsilon^{w})$, where $f^{w}$ denotes the regression
function and $\varepsilon^{w}$ denotes all the information about
the distribution of disturbances in model $w$. Denote the subset
of models satisfying $f\leq0$ almost surely by $\mathcal{G}_{0}$.

\subsection{\label{sub:Size}Size Properties of the Test}

Analysis of size properties of the test is complicated because the
asymptotic distribution of the test statistic is unknown. Instead,
I use a finite sample approach based on the Lindeberg method. For
each sample size $n$, this method gives an upper error bound on approximating
the expectation of smooth functionals of the test statistic by its
expectation calculated assuming Gaussian noise $\{\varepsilon_{i}\}_{i=1}^{n}$.
I also derive a simple lower bound on the growth rate of the pdf of
the test statistic to show that the expectation of smooth functionals
can be used to approximate the expectation of indicator functions.
Combining these results leads to the approximation of the cdf of the
test statistic by its cdf calculated assuming Gaussian disturbances
with an explicit error bound. This allows me to derive certain conditions
which insure that the error converges to zero as the sample size $n$
increases, which is a key step in establishing the bootstrap validity.

The first theorem states that the test has correct asymptotic size
uniformly over the class of models $\mathcal{G}_{0}$ both for plug-in
and RMS test functions. In addition, the test with the plug-in test
function is nonconservative as the size of the test converges to the
required level $\alpha$ uniformly over the class of models $\mathcal{G}_{0}$
with $f^{w}\equiv0_{p}$. When I set $\gamma=\gamma_{n}\rightarrow0$,
the same holds for the test with the RMS test function.
\begin{thm}
\label{thm: size}Let assumptions 1-8 hold. Then for $P=PIA$ or $RMS$,
\[
\inf_{w\in\mathcal{G}_{0}}E_{w}[g_{1-\alpha}^{P}(\hat{T})]\geq1-\alpha+o(1)
\]
In addition, 
\[
\sup_{w\in\mathcal{G}_{0},f^{w}\equiv0_{p}}E_{w}[g_{1-\alpha}^{PIA}(\hat{T})]=1-\alpha+o(1)
\]
 and if $\gamma_{n}\rightarrow0$, then 
\[
\sup_{w\in\mathcal{G}_{0},f^{w}=0_{p}}E[g_{1-\alpha}^{RMS}(\hat{T})]=1-\alpha+o(1)
\]
 as well.
\end{thm}
Proofs of all results are presented in the Appendix. From the proof
of theorem \ref{thm: size}, I also have
\begin{cor}
\label{Corr: deterministic version}If instead of \ref{ass: test function}(ii)
we assume $(\log n)^{19}/(h_{\min}^{3d}n)\rightarrow0$, then theorem
1 holds with $\beta=\beta_{n}=0$.
\end{cor}
The case $\beta=0$ corresponds to the deterministic version of the
test, which rejects the null if and only if $\hat{T}>c_{1-\alpha}^{P}$
for $P=PIA$ or $RMS$. However, I can guarantee that this test maintains
the required size only if $h_{\min}$ converges to zero very slowly
since $(\log n)^{19}$ is a very large number for reasonable sample
sizes.

\subsection{\label{sub:Consistency-Fixed}Consistency Against a Fixed Alternative}

Let me introduce a distance between model $w\in\mathcal{G}$ and the
null hypothesis:
\[
\rho(w,H_{0})=\sup_{i=1,...,\infty;\, m=1,...,p}[f_{m}^{w}(X_{i})]_{+}
\]
For any alternative outside of the set $\Theta_{I}$, $\rho(w,H_{0})>0$.
In this section, I state the result that the test is consistent against
any fixed alternative $w$ with $\rho(w,H_{0})>0$ satisfying assumptions
\ref{ass:Design-points}-\ref{ass:Choice function}. Moreover, I show
that the test is consistent uniformly against alternatives whose distance
from the null hypothesis is bounded away from zero. For $\rho>0$,
let $\mathcal{G}_{\rho}$ denote the subset of all elements of $\mathcal{G}$
such that $\rho(w,H_{0})\geq\rho$ for all $w\in\mathcal{G}_{\rho}$.
Then
\begin{thm}
\label{Theorem: fixed alternatives}Let assumptions 1-8 hold. Then
for $P=PIA$ or $RMS$, 
\[
\sup_{w\in\mathcal{G}_{\rho}}E_{w}[g_{1-\alpha}^{P}(\hat{T})]\rightarrow0
\]
 as $n\rightarrow\infty$.
\end{thm}

\subsection{\label{sub:Consistency-One-Directional}Consistency Against One-Directional
Alternatives}

Let $w(0)\in\mathcal{G}$ be such that $\rho(w(0),H_{0})>0$. For
some sequence $\{a_{n}\}_{n=1}^{\infty}$ of positive numbers converging
to zero, let $f^{n}=a_{n}f^{w(0)}$ be a sequence of local alternatives.
I refer to such sequences as local one-directional alternatives. This
section establishes the consistency of the test against such alternatives
whenever $\sqrt{n/\log n}a_{n}\rightarrow\infty$.
\begin{thm}
\label{Theorem: one directional alternative}Let assumptions 1-8 hold.
Then for $P=PIA$ or $RMS$, 
\[
\sup_{w\in\mathcal{G},f^{w}=f^{n}}E_{w}[g_{1-\alpha}^{P}(\hat{T})]\rightarrow0
\]
 as $n\rightarrow\infty$ if $\sqrt{n/\log n}a_{n}\rightarrow\infty$.\end{thm}
\begin{rem*}
Recall the CMI model from the first example mentioned in the introduction
where $m(X,W,\theta)=\theta\tilde{m}(X,W)$ and $E[\tilde{m}(X,W)|X]>0$
almost surely. The theorem above shows that the test developed in
this paper is consistent against sequences of alternatives $\theta_{0}=\theta_{0,n}$
whenever $\sqrt{n/\log n}\theta_{0,n}\rightarrow\infty$ in this model.
So, my test is consistent against virtually the same set of alternatives
in this model as the test of \citet{AndrewsandShi2010}.
\end{rem*}

\subsection{\label{sub:Uniform-Consistency}Uniform Consistency Against Holder
Smoothness Classes}

In this section, I present the rate of uniform consistency of the
test against the class $\mathcal{F}_{\varsigma}(\tau,L)$ under certain
additional constraints. These additional constraints are needed to
deal with some boundary effects. Let $S=\text{cl}\{X_{i}:\, i\in\mathbb{N}\}$
denote the closure of the infinite set of design points. For any $\vartheta>0$,
let $S_{\vartheta}$ be the subset of $S$ such that for any $x\in S_{\vartheta}$,
the ball with center at $x$ and radius $\vartheta$, $B_{\vartheta}(x)$,
is contained in $S$, i.e. $B_{\vartheta}(x)\subset S$. Denote $\zeta=\min(\varsigma+1,\tau)$.
When $\zeta\leq d$, set $\vartheta=\vartheta_{n}=4\sqrt{d}h_{\min}$.
When $\zeta>d$, set $\vartheta=\vartheta_{n}=4\sqrt{d}(\log n/n)^{1/(2\zeta+d)}$.
Let $\mathbb{N}_{\vartheta_{n}}=\{i\in\mathbb{N}:\, X_{i}\in S_{\vartheta_{n}}\}$.
For any $w\in\mathcal{G}$, let 
\[
\rho_{\vartheta_{n}}(w,H_{0})=\sup_{i\in\mathbb{N}_{\vartheta},\, m=1,...,p}[f_{m}^{w}(X_{i})]_{+}
\]
denote the distance between $w$ and $H_{0}$ over set $S_{\vartheta_{n}}$.
For the next theorem, I will use $\rho_{\vartheta_{n}}$-metric (instead
of $\rho$-metric) to measure the distance between alternatives and
the null hypothesis. Such restrictions are quite common in the literature.
See, for example, \citet{Dumbgen2001} and \citet{LeeandSongandWhang2011}.
Let $\mathcal{G}_{\vartheta}$ be the subset of all elements of $\mathcal{G}$
such that $\inf_{w\in\mathcal{G}_{\vartheta}}\rho_{\vartheta_{n}}(w,H_{0})\geq Ch_{\min}^{\zeta}$
for some large constant $C$ if $\zeta\leq d$ and $\inf_{w\in\mathcal{G}_{\vartheta}}\rho_{\vartheta_{n}}(w,H_{0})(n/\log n)^{\zeta/(2\zeta+d)}\rightarrow\infty$
if $\zeta>d$. Then
\begin{thm}
\label{thm: uniform power}Let assumptions 1-8 hold. For $P=PIA$
or $RMS$, if (i) $\zeta\leq d$ or (ii) $\zeta>d$ and $h_{\min}<(\log n/n)^{1/(2\zeta+d)}$
for large enough $n$, then 
\[
\sup_{w\in\mathcal{G}_{\vartheta}}E_{w}[g_{1-\alpha}^{P}(\hat{T})]\rightarrow0
\]
 as $n\rightarrow\infty$.\end{thm}
\begin{rem*}
Recall the CMI model from the second example mentioned in the introduction
where $m(X,W,\theta)=\tilde{m}(X,W)+\theta$. Assume that $X\in\mathbb{R}$
and $E[\tilde{m}(X,W)|X]=-|X|^{\nu}$ with $\nu>1$. In this model,
the identified set is $\Theta_{I}=\{\theta\in\mathbb{R}:\,\theta\leq0\}$.
The theorem above shows that the test developed in this paper is consistent
against sequences of alternatives $\theta_{0}=\theta_{0,n}$ whenever
$(n/\log n)^{\nu/(2\nu+1)}\theta_{0,n}\rightarrow\infty$. At the
same time, it follows from \citet{Armstrong1}, the test of \citet{AndrewsandShi2010}
is consistent only if $n^{\nu/(2(\nu+1))}\theta_{n,0}\rightarrow\infty$,
so their test has a slower rate of consistency than that developed
in this paper.
\end{rem*}

\subsection{\label{sub:Lower-Bound}Lower Bound on the Minimax Rate of Testing}

In this section, I give a lower bound on the minimax rate of testing.
For $S_{\vartheta}$ defined in the previous section, let $N(h,S_{\vartheta_{n}})$
be the largest $m$ such that there exists $\{x_{1},...,x_{m}\}\subset S_{\vartheta_{n}}$
with $\Vert x_{i}-x_{j}\Vert\geq h$ for all $i,j=1,...,m$ if $i\neq j$.
I will assume that $N(h,S_{\vartheta_{n}})\geq Ch^{-d}$ for all $h\in(0,1)$
and large enough $n$ for some constant $C>0$. This condition holds
almost surely under the conditions of lemma \ref{lemma for support}.
Let $\phi_{n}(Y_{1},...,Y_{n})$ denote a sequence of tests, i.e.
$\phi_{n}(Y_{1},...,Y_{n})$ equals the probability of rejecting the
null hypothesis upon observing sample $Y=(Y_{1},...,Y_{n})$. 
\begin{thm}
Let assumptions 1-8 hold. Assume that (i) $N(h,S_{\vartheta_{n}})\geq Ch^{-d}$
for all $h\in(0,1)$ and large enough $n$ for some constant $C>0$,
(ii) $\varsigma=[\tau]$, and (iii) $r_{n}(n/\log n)^{\tau/(2\tau+d)}\rightarrow0$
as $n\rightarrow\infty$ for some sequence of positive numbers $r_{n}$.
Then for any sequence of tests $\phi_{n}(Y_{1},...,Y_{n})$ with $\sup_{w\in\mathcal{G}_{0}}E_{w}[\phi_{n}(Y_{1},...,Y_{n})]\leq\alpha$,
\[
\lim\sup_{n\rightarrow\infty}\inf_{w\in\mathcal{G},\rho_{\vartheta}(w,H_{0})\geq Cr_{n}}E_{w}[\phi_{n}(Y_{1},...,Y_{n})]\leq\alpha
\]

\end{thm}
Since $\mathcal{F}_{[\tau]}(\tau,L)\subset\mathcal{F}(\tau,L)$, the
same lower bound applies for the class $\mathcal{F}(\tau,L)$ as well.
Comparing this result with theorem 4 shows that the test presented
in this paper is minimax rate optimal if $\zeta=\tau>d$ and $h_{\min}$
is chosen to converge to zero fast enough. When $\zeta=\tau=d$ and
$\beta_{n}$ is set to be constant, the test is rate optimal upto
some logarithmic factors if $h_{\min}$ is chosen to converge to zero
as fast as possible satisfying assumption \ref{ass: test function}.
When $\tau<d$, the test is not rate optimal since the rate of consistency
does not match the lower bound.

\section{Models with Infinitely Many CMI}

In this section, I briefly outline an extention of the test to the
case of infinitely many CMI. Suppose that the parameter $\theta$
is restricted by a countably infinite number of CMI, i.e. $p=\infty$.
As before, I am interested in testing the null hypothesis, $H_{0}$,
that $\theta=\theta_{0}$ against the alternative, $H_{a}$, that
$\theta\neq\theta_{0}$. One possible approach to testing in this
model is to construct a test as described in section \ref{sec:The-Test}
based on some finite subset of CMI assuming that as the sample size
$n$ increases, this subset expands covering all CMI in the asymptotics.
The advantage of the finite sample approach used in this paper is
that it immediately gives certain conditions that insure that such
a test maintain the required size asymptotically. Assume that the
test is based on $K=K_{n}\rightarrow\infty$ inequalities. Then
\begin{cor}
Let assumptions \ref{ass:Design-points}-\ref{ass:bandwidth values},
\ref{ass:The-kernel} and \ref{ass:Choice function} hold. In addition,
assume that (i) $\max_{i=1,...,n}\Vert\hat{\Sigma}_{i}-\Sigma_{i}\Vert_{o}=o_{p}(n^{-\kappa})$
for some $\kappa>0$, (ii) $K_{n}\log n/n^{\kappa/4}\rightarrow0$,
(iii) $\beta=\beta_{n}\leq C$, and (iv) $K_{n}^{6}(\log n)^{4}/(\beta_{n}^{10}h_{\min}^{3d}n)\rightarrow0$
as $n\rightarrow\infty$. Then for $P=PIA$ or $RMS$, 
\[
\inf_{w\in\mathcal{G}_{0}}E_{w}[g_{1-\alpha}^{P}(\hat{T})]\geq1-\alpha+o(1)
\]
as $n\rightarrow\infty$. In addition, 
\[
E_{w}[g_{1-\alpha}^{P}(\hat{T})]\rightarrow0
\]
for any $w\in\mathcal{G}_{\rho}$ with $\rho>0$.
\end{cor}
This corollary shows that the randomized test has correct asymptotic
size both with plug-in and RMS critical values and is consistent against
fixed alternatives outside of the set $\Theta_{I}$. Note that $\kappa$
appearing in condition (i) in this corollary will generally be different
from $\kappa$ used in assumption \ref{ass:variance estimator} because
of increasing number of moment functions. Results concerning the test
with determinstic critical values and local power of the test, with
suitable modifications, can also be easily obtained using arguments
similar to those used in the proofs of corollary \ref{Corr: deterministic version}
and theorems \ref{Theorem: one directional alternative} and \ref{thm: uniform power}.
For brevity, I do not discuss these results.

\section{\label{sec:Monte-Carlo-Results}Monte Carlo Results}

In this section, I present results of Monte Carlo simulations. The
aim of these simulations is twofold. First, I demonstrate that my
test accurately maintain size in finite samples reasonably well. Second,
I compare relative advantages and disadvantages of my test and the
tests of \citet{AndrewsandShi2010}, \citet{ChernozhukovLeeRosen2009},
and \citet{LeeandSongandWhang2011}. The methods of \citet{AndrewsandShi2010}
and \citet{LeeandSongandWhang2011} are most appropriate for detecting
flat alternatives, which represent one-directional local alternatives.
These methods have low power against alternatives with peaks, however.
The test of \citet{ChernozhukovLeeRosen2009} has higher power against
such alternatives, but it requires knowing smoothness properties of
the moment functions. The authors suggest certain rule-of-thumb techniques
to choose a bandwidth value. Finally, the main advantage of my test
is its adaptiveness. In comparison with \citet{AndrewsandShi2010}
and \citet{LeeandSongandWhang2011}, my test has higher power against
alternatives with peaks. In comparison with \citet{ChernozhukovLeeRosen2009},
my test has higher power when their rule-of-thumb techniques lead
to an inappropriate bandwidth value. For example, this happens when
the underlying regression function is mostly flat but varies significantly
in the region where the null hypothesis is violated (the case of spatially
inhomogeneous alternatives, see \citet{LepskiSpokoiny}).

The data generating process in the experiments is
\[
Y=L(M-|X|)_{+}-m+\varepsilon
\]
where $X$, $Y$, and $\varepsilon$ are scalar random variables and
$L$, $M$, and $m$ are some constants. $X$ is distributed uniformly
on $(-2,2)$. Depending on the experiment, $\varepsilon$ is distributed
according to $0.1\cdot N(0,1)$ or $(\xi\cdot0.07+(1-\xi)\cdot0.18)\cdot N(0,1)$
where $\xi$ is a Bernoilly random variable with $p(\xi=1)=0.8$ and
$p(\xi=0)=0.2$ independent of $N(0,1)$. In both cases, $\varepsilon$
is independent of $X$. I consider the following specifications for
parameters. Case 1: $L=M=m=0$. Case 2: $L=0.1$, $M=0.2$, $m=0.02$.
Case 3: $L=M=0$, $m=-0.02$. Case 4: $L=2$, $M=0.2$, $m=0.2$.
Note that $E[Y|X]\leq0$ almost surely in cases 1 and 2 while $P\{E[Y|X]>0\}>0$
in cases 3 and 4. In case 3, the alternative is flat. In case 4, the
alternative has a peak in the region where the null hypothesis is
violated. I have chosen parameters so that rejection probabilities
are strictly greater than 0 and strictly smaller than 1 in most cases
so that meaningful comparisons are possible. I generate samples $(X_{i},Y_{i})_{i=1}^{n}$
of size $n=250$ and $500$ from the distribution of $(X,Y)$. In
all cases, I consider tests with the nominal size $10\%$. The results
are based on 1000 simulations for each specification.

For the test of \citet{AndrewsandShi2010}, I consider their Kolmogorov-Smirnov
test statistic with boxes and truncation parameter $0.05$. I simulate
both plugin (AS, plugin) and GMS (AS, GMS) critical values based on
the bootstrap suggested in their paper. I use the support of the empirical
distribution of $X$ to choose a set of weighting functions. All other
tuning parameters are set as prescribed in their paper. Implementing
all other tests requires selecting a kernel function. In all cases,
I use the following kernel function
\[
K(x)=1.5(1-4x^{2})_{+}
\]
For the test of \citet{ChernozhukovLeeRosen2009}, I use their kernel
type test statistic with critical values based on the multiplier bootstrap
both with (CLR, $\hat{V}$) and without (CLR, $V$) the set estimation.
Both \citet{ChernozhukovLeeRosen2009} and \citet{LeeandSongandWhang2011}
(LSW) circumvent edge effects of kernel estimators by restricting
their test statistics to the proper subsets of the support of $X$.
So, I select 10 and 90\% quantiles of the empirical distribution of
$X$ as bounds for the set over which the test statistics are calculated.
Both tests are nonadaptive. In particular, there is no formal theory
on how to choose bandwidth values in their tests. I use their suggestions
to choose bandwidth values. For the test of \citet{LeeandSongandWhang2011},
I use their test statistic based on one-sided $L_{1}$-norm.

Let me now describe the choice of parameters for the test developed
in this paper. The largest bandwidth value, $h_{\max}$, is set to
be one half of the length of the support of the empirical distribution.
I choose the smallest bandwidth value, $h_{\min}$, so that the kernel
estimator uses on average 15 data points when $n=250$ and 20 data
points when $n=500$. The scaling parameter, $a$, equals $0.8$ so
that the set of bandwidth values is 
\[
H_{n}=\{h=h_{\max}0.8^{k}:\, h\geq h_{\min},k=0,1,2,...\}
\]
My test requires choosing the set $S_{n}$. For each bandwidth value,
$h$, I select the largest subset, $S_{n,h}$, of $X_{i}$'s such
that $X_{i}-X_{j}\geq h$ for any nonequal elements in $S_{n,h}$,
and the smallest $X_{i}$ is always in $S_{n,h}$. Then $S_{n}=\{(i,h):\, h\in H_{n},\, X_{i}\in S_{n,h}\}$.
In all cases, I set $\beta=0$ so that the deterministic version of
the critical values is used. Finally, for the RMS critical value,
I set $\gamma=0.1/\log(n)$ to make meaningful comparisons with the
test of \citet{ChernozhukovLeeRosen2009}. In all bootstrap procedures,
for all tests, I use $1000$ repetitions when $n=250$ and $500$
repetions when $n=500$.

The results of the experiments are presented in table 1 for $n=250$
and in table 2 for $n=500$. In both tables, my test is denoted as
Adaptive test with plug-in and RMS critical values. Consider first
results for $n=250$. In case 1, where the null hypothesis holds,
all tests have rejecting probabilities close to the nominal size 10\%
both for normal and mixture of normals disturbances. In particular,
RMS procedure for my test, GMS procedure for the test of \citet{AndrewsandShi2010}
and the test of \citet{ChernozhukovLeeRosen2009} with the set estimation
do not overreject, which might be concerned based on the construction
of these tests. In case 2, where the null hypothesis holds but the
underlying regression function is mainly strictly below the borderline,
all tests are conservative. When the null hypothesis is violated with
a flat alternative (case 3), the tests of \citet{AndrewsandShi2010}
and \citet{LeeandSongandWhang2011} have highest rejection probabilities
as expected from the theory. In this case, my test is less powerful
in comparison with these tests and somewhat similar to the method
of \citet{ChernozhukovLeeRosen2009}. This is compensated in case
4 where the null hypothesis is violated with the peak-shaped alternative.
In this case, the power of my test is much higher than that of competing
tests. This is especially true for my test with RMS critical values
whose rejection probability exceeds 80\% while rejection probabilities
of competing tests do not exceed 20\%. Note that all results are stable
across distributions of disturbances. Also note that my test with
RMS critical values has much higher power than the test with plugin
critical values in case 4. So, among these two tests, I recommend
the test with RMS critical values. Results for $n=500$ indicate a
similar pattern. Concluding this section, I note that all simulation
results are consistent with the presented theory.

\begin{table}
\caption{Results of Monte Carlo Experiments, $n=250$}

\begin{tabular}{ccccccc>{\centering}p{1.5cm}>{\centering}p{1.5cm}}
\hline 
 &  & \multicolumn{7}{c}{{\scriptsize Probability of Rejecting Null Hypothesis}}\tabularnewline
\cline{3-9} 
{\scriptsize Distribution $\varepsilon$} & {\scriptsize Case} & {\scriptsize AS, plugin} & {\scriptsize AS, GMS} & {\scriptsize LSW} & {\scriptsize CLR, $V$} & {\scriptsize CLR, $\hat{V}$} & {\scriptsize Adaptive test, plugin} & {\scriptsize Adaptive test, RMS}\tabularnewline
\hline 
\multirow{4}{*}{{\scriptsize Normal}} & {\scriptsize 1} & {\scriptsize 0.099} & {\scriptsize 0.102} & {\scriptsize 0.124} & {\scriptsize 0.151} & {\scriptsize 0.151} & {\scriptsize 0.101} & {\scriptsize 0.101}\tabularnewline
 & {\scriptsize 2} & {\scriptsize 0.002} & {\scriptsize 0.007} & {\scriptsize 0.000} & {\scriptsize 0.008} & {\scriptsize 0.008} & {\scriptsize 0.009} & {\scriptsize 0.009}\tabularnewline
 & {\scriptsize 3} & {\scriptsize 0.910} & {\scriptsize 0.910} & {\scriptsize 0.941} & {\scriptsize 0.808} & {\scriptsize 0.808} & {\scriptsize 0.723} & {\scriptsize 0.723}\tabularnewline
 & {\scriptsize 4} & {\scriptsize 0.000} & {\scriptsize 0.143} & {\scriptsize 0.000} & {\scriptsize 0.122} & {\scriptsize 0.191} & {\scriptsize 0.589} & {\scriptsize 0.821}\tabularnewline
\multirow{4}{*}{{\scriptsize Mixture}} & {\scriptsize 1} & {\scriptsize 0.078} & {\scriptsize 0.086} & {\scriptsize 0.107} & {\scriptsize 0.134} & {\scriptsize 0.134} & {\scriptsize 0.124} & {\scriptsize 0.124}\tabularnewline
 & {\scriptsize 2} & {\scriptsize 0.002} & {\scriptsize 0.002} & {\scriptsize 0.000} & {\scriptsize 0.010} & {\scriptsize 0.010} & {\scriptsize 0.016} & {\scriptsize 0.016}\tabularnewline
 & {\scriptsize 3} & {\scriptsize 0.904} & {\scriptsize 0.905} & {\scriptsize 0.925} & {\scriptsize 0.833} & {\scriptsize 0.833} & {\scriptsize 0.692} & {\scriptsize 0.692}\tabularnewline
 & {\scriptsize 4} & {\scriptsize 0.000} & {\scriptsize 0.121} & {\scriptsize 0.000} & {\scriptsize 0.111} & {\scriptsize 0.197} & {\scriptsize 0.555} & {\scriptsize 0.808}\tabularnewline
\hline 
\end{tabular}
\end{table}

\begin{table}

\caption{Results of Monte Carlo Experiments, $n=500$}

\begin{tabular}{ccccccc>{\centering}p{1.5cm}>{\centering}p{1.5cm}}
\hline 
 &  & \multicolumn{7}{c}{{\scriptsize Probability of Rejecting Null Hypothesis}}\tabularnewline
\cline{3-9} 
{\scriptsize Distribution $\varepsilon$} & {\scriptsize Case} & {\scriptsize AS, plugin} & {\scriptsize AS, GMS} & {\scriptsize LSW} & {\scriptsize CLR, $V$} & {\scriptsize CLR, $\hat{V}$} & {\scriptsize Adaptive test, plugin} & {\scriptsize Adaptive test, RMS}\tabularnewline
\hline 
\multirow{4}{*}{{\scriptsize Normal}} & {\scriptsize 1} & {\scriptsize 0.095} & {\scriptsize 0.104} & {\scriptsize 0.119} & {\scriptsize 0.126} & {\scriptsize 0.126} & {\scriptsize 0.103} & {\scriptsize 0.103}\tabularnewline
 & {\scriptsize 2} & {\scriptsize 0.000} & {\scriptsize 0.001} & {\scriptsize 0.000} & {\scriptsize 0.002} & {\scriptsize 0.002} & {\scriptsize 0.008} & {\scriptsize 0.008}\tabularnewline
 & {\scriptsize 3} & {\scriptsize 0.997} & {\scriptsize 0.997} & {\scriptsize 0.996} & {\scriptsize 0.954} & {\scriptsize 0.954} & {\scriptsize 0.903} & {\scriptsize 0.903}\tabularnewline
 & {\scriptsize 4} & {\scriptsize 0.008} & {\scriptsize 0.587} & {\scriptsize 0.000} & {\scriptsize 0.497} & {\scriptsize 0.694} & {\scriptsize 0.976} & {\scriptsize 0.999}\tabularnewline
\multirow{4}{*}{{\scriptsize Mixture}} & {\scriptsize 1} & {\scriptsize 0.120} & {\scriptsize 0.123} & {\scriptsize 0.130} & {\scriptsize 0.117} & {\scriptsize 0.117} & {\scriptsize 0.119} & {\scriptsize 0.119}\tabularnewline
 & {\scriptsize 2} & {\scriptsize 0.000} & {\scriptsize 0.001} & {\scriptsize 0.000} & {\scriptsize 0.000} & {\scriptsize 0.000} & {\scriptsize 0.010} & {\scriptsize 0.010}\tabularnewline
 & {\scriptsize 3} & {\scriptsize 0.993} & {\scriptsize 0.993} & {\scriptsize 0.996} & {\scriptsize 0.949} & {\scriptsize 0.949} & {\scriptsize 0.903} & {\scriptsize 0.903}\tabularnewline
 & {\scriptsize 4} & {\scriptsize 0.005} & {\scriptsize 0.549} & {\scriptsize 0.000} & {\scriptsize 0.456} & {\scriptsize 0.625} & {\scriptsize 0.978} & {\scriptsize 0.997}\tabularnewline
\hline 
\end{tabular}
\end{table}

\section{\label{sec:Conclusions}Conclusions}

In this paper, I developed a new test of conditional moment inequalities.
In contrast to some other tests in the literature, my test is directed
against general nonparametric alternatives, which gives high power
in a large class of CMI models. Considering kernel estimates of moment
functions with many different values of the bandwidth parameter allows
me to construct a test that automatically adapts to the unknown smoothness
of moment functions and selects the most appropriate testing bandwidth
value. The test developed in this paper has uniformly correct asymptotic
size, no matter whether the model is identified, weakly identified,
or not identified, and is uniformly consistent against certain, but
not all, large classes of smooth alternatives whose distance from
the null hypothesis converges to zero at a fastest possible rate.
The tests of \citet{AndrewsandShi2010} and \citet{LeeandSongandWhang2011}
have nontrivial power against $n^{-1/2}$-local one-directional alternatives
whereas my method only allows for nontrivial testing against $(n/\log n)^{-1/2}$-local
alternatives of this type. Additional $(\log n)^{1/2}$ factor should
be regarded as a price for having fast rate of uniform consistency.
There exist sequences of local alternatives against which their tests
are not consistent whereas mine is. Monte Carlo experiments give an
example of a CMI model where finite sample power of my test greatly
exceeds that of competing tests.

\appendix

\section{Appendix}

This Appendix contains proofs of all results stated in the main part
of the paper. Section \ref{sub: equivalence} explains the equivalent
representations for the randomized test. Section \ref{sub:lemma on square root operator}
derives a bound on the modulus of continuity in the operator norm
of the square root operator on the space of symmetric positive semidefinite
matrices. Section \ref{sub:Invariance principle} gives a straighforward
generalization of results in \citet{Chatterjee2005} to the case of
multidimensional random variables. They are concerned with conditions
when the distribution of some function of several independent random
variables with unknown distributions can be approximated by substituting
Gaussian distributions with the same first two moments. They are based
on the Linderberg's argument. The result is specialized to the situation
when the function of interest can be written in the form of the maximum
of linear functions of the data. These results have their own value
as they can be used as an alternative to results on stochastic approximation
from empirical process theory. They are also useful because they give
an explicit bound on the approximation error. Section \ref{sub:Primitive-Conditions}
gives sufficient conditions for assumption \ref{ass:Design-points}
in the main part of the paper. Section \ref{sub:Anticoncentration-Inequality}
presents an anticoncentration inequality for the maximum of Gaussian
random variables with unit variance. Section \ref{sub:Result-on-Gaussian RVs}
describes a result on Gaussian random variables which is used in the
proof of lower bound on the minimax rate. Section \ref{sub:Preliminary-Technical-Results}
develops some preliminary technical results necessary for the proofs
of the main theorems. Finally, section \ref{sub:Proofs-of-Theorems}
presents the proofs of the theorems stated in the main part of the
paper.

Note that all convergence results proven in this Appendix hold uniformly
over the class of models $\mathcal{G}.$ This fact will not be stated
seperately in each special case, but it is assumed everywhere in this
Appendix.

\subsection{\label{sub: equivalence}Lemma on the equivalent representation of
the test}

The lemma below was used in section \ref{sub:The-Test-Function} to
show that the randomized test is equivalent to the test with the random
critical value.
\begin{lem}
\label{lemma: equivalence of tests}$E[g(\hat{T})]=P\{\hat{T}\leq t_{1-\alpha}\}$.\end{lem}
\begin{proof}
Since $g(\hat{T})\in[0,1]$ almost surely,
\begin{equation}
E[g(\hat{T})]=\int_{0}^{1}P\{g(\hat{T})\geq x\}dx\label{eq:56}
\end{equation}
Given that $U$ is independent of the data and, hence, of $\hat{T}$,
\begin{equation}
\int_{0}^{1}P\{g(\hat{T})\geq x\}dx=P\{g(\hat{T})\geq U\}\label{eq:57}
\end{equation}
Finally, note that $\{g(\hat{T})\geq U\}$ is equivalent to $\{\hat{T}\leq t_{1-\alpha}\}$
so that
\begin{equation}
P\{g(\hat{T})\geq U\}=P\{\hat{T}\leq t_{1-\alpha}\}\label{eq:58}
\end{equation}
Combining (\ref{eq:56}), (\ref{eq:57}), and (\ref{eq:58}) gives
the result.
\end{proof}

\subsection{\label{sub:lemma on square root operator}Continuity of the square
root operator on the set of positive semidefinite matrices}
\begin{lem}
\label{lemma:Square root operator}Let $A$ and $B$ be $p\times p$-dimensional
symmetric positive semidefinite matrices. Then $\Vert A^{1/2}-B^{1/2}\Vert_{o}\leq p^{1/2}\Vert A-B\Vert_{o}^{1/2}$
where $\Vert\cdot\Vert_{o}$ means the spectral norm corresponding
to the Euclidean norm on $\mathbb{R}^{p}$.\end{lem}
\begin{proof}
Let $a_{1},...,a_{p}$ and $b_{1},...,b_{n}$ be orthogonal eigenvectors
of matrices $A$ and $B$ correspondingly. Without loss of generality,
I can and will assume that $\Vert a_{i}\Vert=\Vert b_{i}\Vert=1$
for all $i=1,...,p$ where $\Vert\cdot\Vert$ denotes the Euclidean
norm on $\mathbb{R}^{p}$. Let $\lambda_{1}(A),...,\lambda_{p}(A)$
and $\lambda_{1}(B),...,\lambda_{p}(B)$ be corresponding eigenvalues.
Let $f_{i1},...,f_{ip}$ be coordinates of $a_{i}$ in the basis $(b_{1},...,b_{p})$
for all $i=1,...,p$. Then $\sum_{j=1}^{p}f_{ij}^{2}=1$ for all $i=1,...,p$. 

For any $i=1,...,p$,
\begin{eqnarray*}
\sum_{j=1}^{p}(\lambda_{i}(A)-\lambda_{j}(B))^{2}f_{ij}^{2} & = & \Vert\sum_{j=1}^{p}(\lambda_{i}(A)-\lambda_{j}(B))f_{ij}b_{j}\Vert^{2}\\
 & = & \Vert\lambda_{i}(A)a_{i}-\sum_{j=1}^{p}\lambda_{j}(B)f_{ij}b_{j}\Vert^{2}\\
 & = & \Vert(A-B)a_{i}\Vert^{2}\\
 & \leq & \Vert A-B\Vert_{o}^{2}
\end{eqnarray*}
since $\Vert(A-B)a_{i}\Vert\leq\Vert A-B\Vert_{o}\Vert a_{i}\Vert=\Vert A-B\Vert_{o}$.

For $P=A,B$, $P^{1/2}$ has the same eigenvectors as $P$ with corresponding
eigenvalues equal to $\lambda_{1}^{1/2}(P),...,\lambda_{n}^{1/2}(P)$.
Therefore, for any $i=1,...,p$,
\begin{eqnarray*}
\Vert(A^{1/2}-B^{1/2})a_{i}\Vert^{2} & = & \sum_{j=1}^{p}(\lambda_{i}^{1/2}(A)-\lambda_{j}^{1/2}(B))^{2}f_{ij}^{2}\\
 & \leq & \sum_{j=1}^{p}|\lambda_{i}(A)-\lambda_{j}(B)|f_{ij}^{2}\\
 & \leq & \left(\sum_{j=1}^{p}(\lambda_{i}(A)-\lambda_{j}(B))^{2}f_{ij}^{2}\right)^{1/2}\\
 & \leq & \Vert A-B\Vert_{o}
\end{eqnarray*}
where the last line used the inequality derived above. For any $c\in\mathbb{R}^{p}$
with $\Vert c\Vert=1$, let $d_{1},...,d_{p}$ be coordinates of $c$
in the basis $(a_{1},...,a_{p})$. Then
\begin{eqnarray*}
\Vert(A^{1/2}-B^{1/2})c\Vert & = & \Vert(A^{1/2}-B^{1/2})\sum_{i=1}^{p}d_{i}a_{i}\Vert\\
 & \leq & \sum_{i=1}^{p}|d_{i}|\Vert(A^{1/2}-B^{1/2})a_{i}\Vert\\
 & \leq & \sum_{i=1}^{p}|d_{i}|\Vert A-B\Vert_{o}^{1/2}\\
 & \leq & p^{1/2}\Vert A-B\Vert_{o}^{1/2}
\end{eqnarray*}
since $\sum_{i=1}^{p}d_{i}^{2}=1$. Thus, $\Vert A^{1/2}-B^{1/2}\Vert_{o}\leq p^{1/2}\Vert A-B\Vert_{o}^{1/2}$.
\end{proof}

\subsection{\label{sub:Invariance principle}Invariance principle}

In this section, I generalize results of \citet{Chatterjee2005} to
the case of random vectors ($p>1$). I also specialize results for
the case of linear functions because it allows to greatly improve
some constants in Chatterjee's derivation. Let $Z_{1},...,Z_{n}$
be a sequence of independent $p$-dimensional random vectors with
$E[Z_{j}]=0$ for all $j=1,...,n$. Denote $Z=(Z_{1},...,Z_{n})$.
For each $k=1,...,K$ and $m=1,...,p$, let $f_{km}(Z)=\sum_{j=1}^{n}a_{kjm}Z_{j,m}$
be some linear function of $Z$ where $a_{kjm}\geq0$ for each $k=1,...,K$,
$j=1,...,n$, and $m=1,...,p$, and $Z_{j,m}$ denotes $m$-th component
of vector $Z_{j}$. Let $U_{1},...,U_{n}$ be a sequence of independent
normal $p$-dimensional random vectors such that $E[U_{j}]=0$ and
$E[Z_{j}Z_{j}^{T}]=E[U_{j}U_{j}^{T}]$ for each $j=1,...,n$. Denote
$U=(U_{1},...,U_{n})$ and 
\[
C(g)=\Vert g^{\prime\prime\prime}\Vert_{\infty}+3\Vert g^{\prime\prime}\Vert_{\infty}+\Vert g^{\prime}\Vert_{\infty}
\]
Denote $a=\max_{k,j,m}a_{kjm}$. Then
\begin{thm}
\label{thm:Chatterjee}For any thrice differentiable function $g$
on $\mathbb{R}$, 
\begin{multline*}
E[g(\max_{k,m}f_{km}(Z))]-E[g(\max_{k,m}f_{km}(U))]\leq\\
(3/6^{1/3})pa(C(g)n)^{1/3}(\Vert g^{\prime}\Vert_{\infty}\log(Kp))^{2/3}\{\max_{j,m}E[|Z_{j,m}|^{3}]+\max_{j,m}E[|U_{j,m}|^{3}]\}^{1/3}
\end{multline*}
\end{thm}
\begin{rem*}
The constant in the inequality above can be improved somewhat by using
expressions for $A_{1}$, $A_{2}$, and $A_{3}$ in the proof given
below. I do not follow this step because that would mess up the statement
of the theorem significantly.\end{rem*}
\begin{proof}
As in \citet{Chatterjee2005}, for $\alpha\geq1$, let $F_{\alpha}:\,\mathbb{R}^{p\times n}$
be such that 
\[
F_{\alpha}(x)=\alpha^{-1}\log(\sum_{k,m}\exp(\alpha f_{km}(x)))
\]
 for all $x\in\mathbb{R}^{p\times n}$. Then
\begin{eqnarray*}
\max_{k,m}f_{km}(x) & = & \alpha^{-1}\log(\exp(\alpha\max_{k,m}f_{km}(x)))\\
 & \leq & \alpha^{-1}\log(\sum_{k,m}\exp(\alpha f_{km}(x)))\\
 & \leq & \alpha^{-1}\log(Kp\exp(\alpha\max_{k,m}f_{km}(x)))\\
 & \leq & \alpha^{-1}\log(Kp)+\max_{k,m}f_{km}(x)
\end{eqnarray*}
So,
\[
|\max_{k,m}f_{km}(x)-F_{\alpha}(x)|\leq\alpha^{-1}\log(Kp)
\]
Thus,
\begin{multline*}
|E[g(\max_{k,m}f_{km}(Z))]-E[g(\max_{k,m}f_{km}(U))]|\leq\\
2\Vert g^{\prime}\Vert_{\infty}\alpha^{-1}\log(Kp)+|E[g(F_{\alpha}(Z))]-E[g(F_{\alpha}(U))]|
\end{multline*}
For any $j=0,...,n$, denote $Z^{j}=(Z_{1},...,Z_{j},U_{j+1},...,U_{n})$.
Then
\[
|E[g(F_{\alpha}(Z))]-E[g(F_{\alpha}(U))]|\leq\sum_{j=1}^{n}|E[g(F_{\alpha}(Z^{j})]-E[g(F(Z^{j-1}))]|
\]
For $Z_{1},...,Z_{j-1},U_{j+1},...,U_{n}$ fixed, denote $l(Z_{j})=g(F_{\alpha}(Z^{j}))$.
By Taylor formula,
\begin{eqnarray*}
g(F_{\alpha}(Z^{j})-g(F_{\alpha}(Z^{j-1}) & = & l(Z_{j})-l(U_{j})\\
 & = & \sum_{m_{1}}\frac{\partial l(0)}{\partial Z_{jm_{1}}}(Z_{jm_{1}}-U_{jm_{1}})\\
 & + & (1/2)\sum_{m_{1},m_{2}}\frac{\partial^{2}l(0)}{\partial Z_{jm_{1}}\partial Z_{jm_{2}}}(0)(Z_{jm_{1}}Z_{jm_{2}}-U_{jm_{1}}U_{jm_{2}})\\
 & + & (1/6)\sum_{m_{1},m_{2},m_{3}}\frac{\partial^{3}l(\tilde{Z})}{\partial Z_{jm_{1}}\partial Z_{jm_{2}}\partial Z_{jm_{3}}}Z_{jm_{1}}Z_{jm_{2}}Z_{jm_{3}}\\
 & - & (1/6)\sum_{m_{1},m_{2},m_{3}}\frac{\partial^{3}l(\tilde{U})}{\partial Z_{jm_{1}}\partial Z_{jm_{2}}\partial Z_{jm_{3}}}U_{jm_{1}}U_{jm_{2}}U_{jm_{3}}
\end{eqnarray*}
where $\tilde{Z}$ and $\tilde{U}$ are on the lines connecting $0$
and $Z_{j}$ and $0$ and $U_{j}$ correspondingly. By independence,
\begin{multline*}
|E[g(F_{\alpha}(Z^{j})]-E[g(F(Z^{j-1}))]|\\
\leq(1/6)\sum_{m_{1},m_{2},m_{3}}\sup_{X\in\mathbb{R}^{p\times n}}\left|\frac{\partial^{3}g(F_{\alpha}(X))}{\partial X_{jm_{1}}\partial X_{jm_{2}}\partial X_{jm_{3}}}\right|(E[|Z_{jm_{1}}Z_{jm_{2}}Z_{jm_{3}}|]+E[|U_{jm_{1}}U_{jm_{2}}U_{jm_{3}}|])
\end{multline*}
By Holder inequality, 
\[
E[|Z_{jm_{1}}Z_{jm_{2}}Z_{jm_{3}}|]\leq\max_{m}E[|Z_{jm}|^{3}]
\]
 and 
\[
E[|U_{jm_{1}}U_{jm_{2}}U_{jm_{3}}|]\leq\max_{m}E[|U_{jm}|^{3}]
\]

Denote
\[
A_{1}=\sup_{X\in\mathbb{R}^{p\times n}}\left|\frac{\partial F_{\alpha}(X)}{\partial X_{jm_{1}}}\frac{\partial F_{\alpha}(X)}{\partial X_{jm_{2}}}\frac{\partial F_{\alpha}(X)}{\partial X_{jm_{3}}}\right|
\]
\begin{multline*}
A_{2}=\sup_{X\in\mathbb{R}^{p\times n}}\left|\frac{\partial F_{\alpha}(X)}{\partial X_{jm_{1}}}\frac{\partial^{2}F_{\alpha}(X)}{\partial X_{jm_{2}}\partial X_{jm_{3}}}\right|+\\
\sup_{X\in\mathbb{R}^{p\times n}}\left|\frac{\partial F_{\alpha}(X)}{\partial X_{jm_{2}}}\frac{\partial^{2}F_{\alpha}(X)}{\partial X_{jm_{1}}\partial X_{jm_{3}}}\right|+\sup_{X\in\mathbb{R}^{p\times n}}\left|\frac{\partial F_{\alpha}(X)}{\partial X_{jm_{3}}}\frac{\partial^{2}F_{\alpha}(X)}{\partial X_{jm_{1}}\partial X_{jm_{2}}}\right|
\end{multline*}
and
\[
A_{3}=\sup_{X\in\mathbb{R}^{p\times n}}\left|\frac{\partial^{3}F_{\alpha}(X)}{\partial X_{jm_{1}}\partial X_{jm_{2}}\partial X_{jm_{3}}}\right|
\]
Then
\[
\sup_{X\in\mathbb{R}^{p\times n}}\left|\frac{\partial^{3}g(F_{\alpha}(X))}{\partial X_{jm_{1}}\partial X_{jm_{2}}\partial X_{jm_{3}}}\right|\leq\Vert g^{\prime\prime\prime}\Vert_{\infty}A_{1}+\Vert g^{\prime\prime}\Vert_{\infty}A_{2}+\Vert g^{\prime}\Vert_{\infty}A_{3}
\]
So, it only remains to bound partial derivatives of $F_{\alpha}$. 

To simplify notation, denote $B_{km}=\exp(\alpha f_{km}(X))$ for
$k=1,...,K$ and $m=1,...,p$. Then
\[
\frac{\partial F_{\alpha}(X)}{\partial X_{jm_{1}}}=\frac{\sum_{k}B_{km_{1}}a_{kjm_{1}}}{\sum_{k,m}B_{km}}
\]
The expression on the right hand side of the formula above is the
expectation of a random variable which takes value $a_{kjm_{1}}$
with probability $B_{km_{1}}/\sum_{km}B_{km}$ for $k=1,...,K$ and
$0$ with probability $1-\sum_{k}B_{km_{1}}/\sum_{km}B_{km}$. If
$m_{1}$, $m_{2}$, and $m_{3}$ are all different, then 
\[
\frac{\partial F_{\alpha}(X)}{\partial X_{jm_{1}}}\frac{\partial F_{\alpha}(X)}{\partial X_{jm_{2}}}\frac{\partial F_{\alpha}(X)}{\partial X_{jm_{3}}}
\]
will be the product of expectations of 3 random variables with nonitersecting
supports. It is easy to see that this product will be not greater
than $a^{3}/27$. All other cases can be treated by the same argument.
We have 
\[
A_{1}\leq\begin{cases}
\begin{array}{cc}
a^{3}/27 & \text{if }m_{1}\text{, }m_{2}\text{, and }m_{3}\text{ are all different}\\
4a^{3}/27 & \text{if }m_{1}=m_{2}\neq m_{3}\\
a^{3} & \text{if }m_{1}=m_{2}=m_{3}
\end{array}\end{cases}
\]
If $m_{1}$, $m_{2}$, and $m_{3}$ are all different, then 
\[
\frac{\partial^{2}F_{\alpha}(X)}{\partial X_{jm_{1}}\partial X_{jm_{2}}}=-\alpha\frac{\sum_{k}B_{km_{1}}a_{kjm_{1}}\sum_{k}B_{km_{2}}a_{kjm_{2}}}{(\sum_{km}B_{km})^{2}}
\]
and
\[
\frac{\partial^{3}F_{\alpha}(X)}{\partial X_{jm_{1}}\partial X_{jm_{2}}\partial X_{jm_{3}}}=2\alpha^{2}\frac{\sum_{k}B_{km_{1}}a_{kjm_{1}}\sum_{k}B_{km_{2}}a_{kjm_{2}}\sum_{k}B_{km_{3}}a_{kjm_{3}}}{(\sum_{km}B_{km})^{3}}
\]
If $m_{1}=m_{2}\neq m_{3}$, then
\[
\frac{\partial^{2}F_{\alpha}(X)}{\partial X_{jm_{1}}\partial X_{jm_{2}}}=-\alpha\frac{(\sum_{k}B_{km_{1}}a_{kjm_{1}})^{2}}{(\sum_{km}B_{km})^{2}}+\alpha\frac{\sum_{k}B_{km_{1}}a_{kjm_{1}}^{2}}{\sum_{km}B_{km}}
\]
and
\begin{multline*}
\frac{\partial^{3}F_{\alpha}(X)}{\partial X_{jm_{1}}\partial X_{jm_{2}}\partial X_{jm_{3}}}\\
=2\alpha^{2}\frac{(\sum_{k}B_{km_{1}}a_{kjm_{1}})^{2}\sum_{k}B_{km_{3}}a_{kjm_{3}}}{(\sum_{km}B_{km})^{3}}-\alpha^{2}\frac{\sum_{k}B_{km_{1}}a_{kjm_{1}}^{2}\sum_{k}B_{km_{3}}a_{kjm_{3}}}{(\sum_{km}B_{km})^{2}}
\end{multline*}
If $m_{1}=m_{2}=m_{3}$, then
\begin{multline*}
\frac{\partial^{3}F_{\alpha}(X)}{\partial X_{jm_{1}}\partial X_{jm_{2}}\partial X_{jm_{3}}}\\
=\alpha^{2}\frac{\sum_{k}B_{km_{1}}a_{kjm_{1}}^{3}}{(\sum_{km}B_{km})}-3\alpha^{2}\frac{\sum_{k}B_{km_{1}}a_{kjm_{1}}^{2}\sum_{k}B_{km_{1}}a_{kjm_{1}}}{(\sum_{km}B_{km})^{2}}+2\alpha^{2}\frac{(\sum_{k}B_{km_{1}}a_{kjm_{1}})^{3}}{(\sum_{km}B_{km})^{3}}
\end{multline*}
So,
\[
A_{2}\leq\begin{cases}
\begin{array}{cc}
3\alpha a^{3}/27 & \text{if }m_{1}\text{, }m_{2}\text{, and }m_{3}\text{ are all different}\\
59\alpha a^{3}/108 & \text{if }m_{1}=m_{2}\neq m_{3}\\
3\alpha a^{3} & \text{if }m_{1}=m_{2}=m_{3}
\end{array}\end{cases}
\]
and
\[
A_{3}\leq\begin{cases}
\begin{array}{cc}
2\alpha^{2}a^{3}/27 & \text{if }m_{1}\text{, }m_{2}\text{, and }m_{3}\text{ are all different}\\
8\alpha^{2}a^{3}/27 & \text{if }m_{1}=m_{2}\neq m_{3}\\
\alpha^{2}a^{3} & \text{if }m_{1}=m_{2}=m_{3}
\end{array}\end{cases}
\]
Therefore,
\begin{multline*}
|E[g(\max_{k,m}f_{km}(Z))]-E[g(\max_{k,m}f_{km}(U))]|\\
\leq2\Vert g^{\prime}\Vert_{\infty}\alpha^{-1}\log(Kp)+\frac{np^{3}\alpha^{2}a^{3}}{6}C(g)\left[\max_{j,m}E[|Z_{jm}|^{3}+\max_{j,m}E[|U_{jm}|^{3}]]\right]
\end{multline*}
Optimizing with respect to $\alpha$ yields the result.
\end{proof}

\subsection{\label{sub:Primitive-Conditions}Primitive Conditions for Assumption
1}

In this section, I give a counter-example for the statement that for
assumption \ref{ass:Design-points} to hold, it siffices to assume
that $\{X_{i}:\, i=1,...,n\}$ are sampled from a distribution that
is absolutely continuous with respect to Lebegue measure, has bounded
support, and whose density is bounded from above and away from zero
on the support. I also prove that assumption \ref{ass:Design-points}
holds if, in addition to above conditions, one assumes that the support
is a convex set.
\begin{lem}
\label{lem:wrong statement}There exist a probability distribution
on $[-1,1]^{2}$ which is uniform on its support such that if $\{X_{i}:\, i=1,...,n\}$
are sampled from this distribution, then assumption \ref{ass:Design-points}
fails.\end{lem}
\begin{proof}
As an example of such a probability distribution, consider the uniform
distribution on 
\[
S=\{(x_{1},x_{2})\in[-1,1]^{2}:\, x_{1}\geq0;\,-(1+\alpha)x_{1}^{\alpha}/2\leq x_{2}\leq(1+\alpha)x_{1}^{\alpha}/2\}
\]
for some $\alpha>0$. For fixed $i$, the probability that $X_{i,1}\leq\underline{h}$
is $\underline{p}=\underline{h}^{1+\alpha}$, and the probability
that $X_{i,1}>\overline{h}$ is $\overline{p}=1-\overline{h}^{1+\alpha}$.
Let $A_{n}$ be an event that $X_{i,1}\leq\underline{h}$ for exactly
one $i=1,...,n$ whereas $X_{i,1}>\overline{h}$ for all other $i=1,...,n$
with $\underline{h}<\overline{h}$. The probability of this event
is 
\[
P(A_{n})=n\underline{p}\overline{p}^{n-1}=n\underline{h}^{1+\alpha}(1-\overline{h}^{1+\alpha})^{n-1}
\]
Set $\underline{h}=(C_{1}/n)^{1/(1+\alpha)}$ and $\overline{h}=(C_{2}/n)^{1/(1+\alpha)}$
with $0<C_{1}<C_{2}<1$. Then we can find the limit of $P(A_{n})$
as $n\rightarrow\infty$:
\[
\lim_{n\rightarrow\infty}P(A_{n})=\lim_{n\rightarrow\infty}C_{1}(1-C_{2}/n)^{n-1}=C_{1}e^{-C_{2}}>0
\]
Note that on $A_{n}$, there is an observation $X_{i}$ such that
there is no other observations in the ball with center at $X_{i}$
and radius $(C_{2}^{1/(1+\alpha)}-C_{1}^{1/(1+\alpha)})/n^{1/(1+\alpha)}$.
The result now follows by choosing $\alpha$ sufficiently large such
that $n^{-1/(1+\alpha)}$ converges to zero slower then $h_{\min}$.
\end{proof}
Now I give a sufficient primitive condition for assumption \ref{ass:Design-points}.
\begin{lem}
\label{lemma for support}If $\{X_{i}:\, i=1,...,n\}$ are sampled
from a distribution which is absolutely continuous with respect to
Lebegue measure, has bounded and convex support $S\subset\mathbb{R}^{d}$,
and whose density is bounded from above and away from zero on the
support, then assumption \ref{ass:Design-points} holds for large
$n$ almost surely.\end{lem}
\begin{proof}
Consider sets of the following form: $I(a_{1},...,a_{d},c)=S\cap\{x:\, a_{1}x_{1}+...+a_{d}x_{d}=c\}$
with $a_{1}^{2}+...+a_{d}^{2}=1$. These are convex sets. It follows
from the fact that the density is bounded from above that $\inf_{a_{1},...,a_{d}}\sup_{c}D(I(a_{1},...,a_{d},c))>0$
where $D(\cdot)$ denotes the diameter of the set. So, there exists
some constant $0<C\leq1$ such that for all $r<1$ and all $x\in S$,
each ball with center at $x$ and radius $r$ has at least fraction
$C$ of its Lebegue measure inside of the support $S$: $\lambda(B(x,r)\cap S)/\lambda(B(x,r))>C$. 

Note that $\delta$-covering numbers of the set $S$ satisfy $N(\delta)\lesssim\delta^{d}$
as $\delta\rightarrow0$, i.e. there exists some constant $C>0$ such
that $N(\delta,S)<C/\delta^{d}$. Consider the lower bound. For each
$h\in H_{n}$, consider the set of covering balls with centers $G_{h,1}$,...,$G_{h,N(h)}$
and radii $\delta_{h}=h/2$. Then for each $X_{i}$ and $h\in H_{n}$,
there exists some $j\in\{1,...,N(h)\}$ such that $B(X_{i},h)\supset B(G_{h,j},\delta_{h})$.
Thus, it is enough to prove the lower bound for the number of observations
droping into these covering balls. Since the density is bounded away
from zero, there exists some constant $C>0$ such that for each $h\in H_{n}$
and $j=1,...,N(h)$, $P(X_{i}\in B(G_{h,j},\delta_{h}))>Ch^{d}$.
Denote $I_{h,j}(X_{i})=I\{X_{i}\in B(G_{h,j},\delta_{h})\}$. A Hoeffding
inequality (see proposition 1.3.5 in \citet{Dudley1999}) gives
\[
P\{\sum_{i=1}^{n}I_{h,j}(X_{i})/n<Ch^{d}/2\}\leq P\{\sum_{i=1}^{n}I_{h,j}(X_{i})/n-E[I_{h,d}(X_{i})]<-Ch^{d}/2)\}\leq C\exp(-Cnh^{d})
\]
Then by union bound,
\[
P(\cup_{h\in H_{n},j=1,...,N(h)}\{\sum_{i=1}^{n}I_{h,j}(X_{i})/n<Ch^{d}/2\})\leq Ch_{\min}^{-d}\log n\exp(-Cnh_{\min}^{d})\rightarrow0
\]
as $n\rightarrow\infty$. Summing the probabilities above over $n$,
we conclude, by the Borel-Cantelli lemma, that the lower bound in
assumption \ref{ass:Design-points}(iii) holds for large $n$ almost
surely. A similar argument gives the upper bound.
\end{proof}

\subsection{\label{sub:Anticoncentration-Inequality}Anticoncentration Inequality
for the Maximum of Gaussian Random Variables}

In this section, I derive an upper bound for the pdf of the maximum
of correlated Gaussian random variables satisfying certain assumptions.
Let $\{Z_{i}:\, i=1,...,S\}$ be a set of standard Gaussian random
variables. Assume that this set contains at least $M$ independent
random variables. Define $W=\max_{i=1,...,S}Z_{i}$. Let $m$ denote
the median of $W$ and $f_{W}(\cdot)$ denote its pdf. Then
\begin{lem}
\label{anticoncentration}$\sup_{w>m}f_{W}(w)\leq C\sqrt{\log(M+1)}S/M$
for some universal constant $C$.\end{lem}
\begin{proof}
The case $M=1$ is trivial. So, assume that $M>1$. Let $\Phi(\cdot)$
and $\phi(\cdot)$ denote the cdf and the pdf of the standard Gaussian
distribution. Since there is at least $M$ independent standard Gaussian
random variables, $\Phi^{M}(m)\geq1/2$ and $m>0$. So, there exists
some constant $C>0$ such that $\Phi(x)\leq1-\phi(x)/(Cx)$ for any
$x\geq m$ (see proposition 2.2.1 in \citet{Dudley1999})and
\[
\left(1-\frac{\phi(m)}{Cm}\right)^{M}\geq\frac{1}{2}
\]
Let $y$ denote the unique positive real number such that
\begin{equation}
\left(1-\frac{\phi(y)}{Cy}\right)^{M}=\frac{1}{2}\label{eq: normal bound}
\end{equation}
Note that $y\leq m$. In addition, $y$ is increasing in $M$, so
there exists some constant $C_{1}>0$ such that $y>C_{1}$ for any
$M\geq2$. Taking logs of both sides of equation (\ref{eq: normal bound})
and noting that $\log(1+x)\leq x$ for any $x\in\mathbb{R}$, we obtain
$\phi(y)\leq y\log C/M$ for some constant $C>1$. On the other hand,
$\phi(y)/(Cy)<1/2$. So inequality $\log(1+x)\geq2x$ for any $x\in(-1/2,0]$
gives
\[
\frac{2M\phi(y)}{Cy}\geq\log2
\]
Combining this inequality with $y>C_{1}$ yields $y\leq C\sqrt{\log(M+1)}$
for any $M$ if $C$ is sufficiently large. Therefore, $\phi(y)\leq C\sqrt{\log(M+1)}/M$
and for any $w>m$, 
\[
f_{W}(w)\leq S\phi(w)\leq S\phi(m)\leq S\phi(y)\leq C\sqrt{\log(M+1)}S/M
\]

\end{proof}

\subsection{\label{sub:Result-on-Gaussian RVs}Result on Gaussian Random Variables}

In this section, I state a result on Gaussian random variables which
will be used in the derivation of the lower bound on the rate of uniform
consistency.
\begin{lem}
\label{property of Gauss rv}Let $\xi_{n}$, $n=1,...,\infty$, be
a sequence of independent standard Gaussian random variables and $w_{i,n}$,
$i=1,...,n$, $n=1,...,\infty$, be a triangular array of positive
numbers. If $w_{i,n}<C\sqrt{\log n}$ with $C\in(0,1)$ for all $i=1,...,n$,
$n=1,...,\infty$, then
\[
\lim_{n\rightarrow\infty}E[|n^{-1}\sum_{i=1}^{n}\exp(w_{i,n}\xi_{i}-w_{i,n}^{2}/2)-1|]=0
\]
\end{lem}
\begin{proof}
The proof is based on the generalization of lemma 6.2 in \citet{Dumbgen2001}.
Denote $Z_{i,n}=\exp(w_{i,n}\xi_{i}-w_{i,n}^{2}/2)$ and $t_{n}=(E[\sum_{i=1}^{n}Z_{i,n}/n-1]^{2})^{1/2}$.
Note that $EZ_{i,n}=1$ and $EZ_{i,n}^{2}=\exp(w_{i,n}^{2})$. Thus,
\[
t_{n}^{2}=(\sum_{i=1}^{n}(EZ_{i,n}^{2}-(EZ_{i,n})^{2}))/n^{2}\leq\sum_{i=1}^{n}\exp(w_{i,n}^{2})/n^{2}\rightarrow0
\]
if $\max_{i=1,...,n}\exp(w_{i,n}^{2})/n\rightarrow0$. The last condition
holds by assumption. So,
\begin{eqnarray*}
E|n^{-1}\sum_{i=1}^{n}\exp(w_{i,n}\xi_{i}-w_{i,n}^{2}/2)-1| & = & \int_{0}^{\infty}P(|n^{-1}\sum_{i=1}^{n}Z_{i,n}-1|>t)dt\\
 & \leq & t_{n}+\int_{t_{n}}^{\infty}t_{n}^{2}/t^{2}dt\\
 & = & 2t_{n}\rightarrow0
\end{eqnarray*}

\end{proof}

\subsection{\label{sub:Preliminary-Technical-Results}Preliminary Technical Results}

In this section, I derive some necessary preliminary results that
are used in the proofs of the theorems stated in the main part of
the paper. It is assumed throughout that assumptions \ref{ass:Design-points}-\ref{ass:Choice function}
hold. I will use the following additional notation. Let $\{\psi_{n}\}_{n=1}^{\infty}$
be a sequence of positive real numbers such that $\psi_{n}\geq C_{\psi}p\log n/n^{\kappa/4}$
for some large constant $C_{\psi}>0$ and $\psi_{n}\rightarrow0$
as $n\rightarrow\infty$. For any $\lambda\in(0,1)$, define $c_{1-\lambda}^{PIA,0}\in\mathbb{R}$
and $g_{1-\lambda}^{PIA,0}:\,\mathbb{R}\rightarrow[0,1]$ by analogy
with $c_{1-\lambda}^{PIA}$ and $g_{1-\lambda}^{PIA}$ with $\Sigma_{i}$
used instead of $\hat{\Sigma}_{i}$ for all $i=1,...,n$. Denote $S_{n}^{D}=\{s\in S_{n}:\, f_{s}/V_{s}>-(c_{1-\gamma_{n}-\psi_{n}}^{PIA,0}+\beta_{n})\}$.
For any $\lambda\in(0,1)$, define $c_{1-\lambda}^{D}\in\mathbb{R}$
and $g_{1-\lambda}^{D}:\,\mathbb{R}\rightarrow[0,1]$ by analogy with
$c_{1-\lambda}^{RMS}$ and $g_{1-\lambda}^{RMS}$ with $S_{n}^{D}$
used instead of $S_{n}^{RMS}$. Let $\{\epsilon_{i}:\, i=1,...,n\}$
be an iid sequence of $p$-dimensional standard Gaussian random vectors
that are independent of the data. Denote $\hat{e}_{j}=\hat{\Sigma}^{1/2}\epsilon_{j}$
and $e_{j}=\Sigma^{1/2}\epsilon_{j}.$ Note that $\hat{e}_{j}$ is
equal in distribution to $\tilde{Y}_{j}$. Finally, denote 
\[
\varepsilon_{i,m,h}=\sum_{j=1}^{n}w_{h}(X_{i},X_{j})\varepsilon_{j,m}
\]
\[
f_{i,m,h}=\sum_{j=1}^{n}w_{h}(X_{i},X_{j})f_{m}(X_{j})
\]
\[
e_{i,m,h}=\sum_{j=1}^{n}w_{h}(X_{i},X_{j})e_{j}
\]
\[
\hat{e}_{i,m,h}=\sum_{j=1}^{n}w_{h}(X_{i},X_{j})\hat{e}_{j}
\]

\[
T^{PIA}=\max_{s\in S_{n}}(\hat{e}_{s}/\hat{V}_{s})
\]
\[
T^{PIA,0}=\max_{s\in S_{n}}(e_{s}/V_{s})
\]
Note that $T^{PIA}$ is equal in distribution to the simulated statistic.

I start with a result on bounds for weights and variances of the kernel
estimator. The same result can be found in \citet{Horowitz2001}.
\begin{lem}
\label{HS}There exist constants $C>0$ and $0<C_{1}<C_{2}<\infty$
such that, for any $i,j=1,...,n$, $m=1,...,p$, and $h\in H_{n}$,
\[
w_{h}(X_{i},X_{j})\leq C/(nh^{d})
\]
and
\[
C_{1}/\sqrt{nh^{d}}\leq V_{i,m,h}\leq C_{2}/\sqrt{nh^{d}}
\]
\end{lem}
\begin{proof}
By assumptions \ref{ass:Design-points} and \ref{ass:The-kernel},
for any $i=1,...,n$ and $h\in H_{n}$,
\[
C_{1}nh^{d}\leq CM_{h/2}(X_{i})\leq\sum_{k=1}^{n}K(X_{i}-X_{k})\leq M_{h}(X_{i})\leq C_{2}nh^{d}
\]
and 
\[
C_{1}nh^{d}\leq\sum_{k=1}^{n}K^{2}(X_{i}-X_{k})\leq C_{2}nh^{d}
\]
for some constants $C>0$ and $0<C_{1}<C_{2}<\infty$. In addition,
$K(X_{i}-X_{j})\leq1$ for any $j=1,...,n$. So,
\[
w_{h}(X_{i}-X_{j})=K(X_{i}-X_{j})/\sum_{k=1}^{n}K(X_{i}-X_{k})\leq C/(nh^{d})
\]
By assumption \ref{ass:Disturbances}, since $\sum_{j=1}^{n}w_{h}(X_{i},X_{j})=1$,
\begin{eqnarray*}
V_{i,m,h} & = & \left(\sum_{j=1}^{n}w_{h}^{2}(X_{i},X_{j})\Sigma_{j,mm}\right)^{1/2}\\
 & \leq & C\left(\sum_{j=1}^{n}w_{h}^{2}(X_{i},X_{j})\right)^{1/2}\\
 & \leq & C\max_{j=1,...,n}w_{h}^{1/2}(X_{i},X_{j})\\
 & \leq & C/\sqrt{nh^{d}}
\end{eqnarray*}
and
\[
V_{i,m,h}\geq C\left(\sum_{j=1}^{n}w_{h}^{2}(X_{i},X_{j})\right)^{1/2}\geq(C/nh^{d})\left(\sum_{j=1}^{n}K^{2}(X_{i}-X_{j})\right)^{1/2}\geq C/\sqrt{nh^{d}}
\]
\end{proof}
\begin{lem}
\label{normal conv}$E[\max_{s\in S_{n}}|e_{s}/V_{s}|]\leq C(\log n)^{1/2}$.\end{lem}
\begin{proof}
For any $s\in S_{n}$, $e_{s}/V_{s}$ is a standard Gaussian random
variable. Denote $\psi=\exp(x^{2})-1$. Let $\Vert\cdot\Vert_{\psi}$
denote $\psi$-Orlicz norm. It is easy to check that $\Vert e_{s}/V_{s}\Vert_{\psi}<C<\infty$.
So, by lemma 2.2.2 in \citet{VaartWellner1996}, 
\[
E[\max_{s\in S_{n}}|e_{s}/V_{s}|]\leq C\Vert\max_{s\in S_{n}}|e_{s}/V_{s}|\Vert_{\psi}\leq C(\log n)^{1/2}
\]
since $|S_{n}|\leq Cn^{\phi}$ for some $\phi>0$.\end{proof}
\begin{lem}
\label{variance}$\max_{s\in S_{n}}|\hat{V}_{s}/V_{s}-1|=o_{p}(n^{-\kappa})$
and $\max_{s\in S_{n}}|V_{s}/\hat{V}_{s}-1|=o_{p}(n^{-\kappa})$.\end{lem}
\begin{proof}
By assumption \ref{ass:Disturbances}, for any $(i,m,h)\in S_{n}$,
\[
V_{i,m,h}^{2}=\sum_{j=1}^{n}w_{h}^{2}(X_{i},X_{j})\Sigma_{j,mm}\geq C\sum_{j=1}^{n}w_{h}^{2}(X_{i},X_{j})
\]
In addition, 
\[
|\hat{V}_{i,m,h}^{2}-V_{i,m,h}^{2}|\leq\sum_{j=1}^{n}w_{h}^{2}(X_{i},X_{j})|\hat{\Sigma}_{j,mm}-\Sigma_{j,mm}|
\]
So,
\begin{eqnarray*}
\max_{s\in S_{n}}|\hat{V}_{s}^{2}/V_{s}^{2}-1| & \leq & C\max_{m=1,...,p}\max_{j=1,...,n}|\hat{\Sigma}_{j,mm}-\Sigma_{j,mm}|\\
 & \leq & C\max_{j=1,...,n}\Vert\hat{\Sigma}_{j}-\Sigma_{j}\Vert_{o}
\end{eqnarray*}
Assumption \ref{ass:variance estimator} gives $\max_{j=1,...,n}\Vert\hat{\Sigma}_{j}-\Sigma_{j}\Vert_{o}=o_{p}(n^{-\kappa})$.
So, $\max_{s\in S_{n}}|\hat{V}_{s}^{2}/V_{s}^{2}-1|=o_{p}(n^{-\kappa})$.
Combining this result with inequality $|x-1|\leq|x^{2}-1|$, which
holds for any $x>0$, yields the first result of the lemma. The second
result follows from the first one and the inequality $|1/x-1|<2|x-1|$,
which holds for any $|x-1|<1/2$.\end{proof}
\begin{lem}
\label{lemma: quantile approximation}$P\{c_{1-\nu_{n}-\psi_{n}}^{PIA,0}>c_{1-\nu_{n}}^{PIA}\}=o(1)$
and $P\{c_{1-\nu_{n}+\psi_{n}}^{PIA,0}<c_{1-\nu_{n}}^{PIA}\}=o(1)$
for any sequences $\{\nu_{n}\}_{n=1}^{\infty}$ and $\{\psi_{n}\}_{n=1}^{\infty}$
of positive numbers satisfying $\nu_{n}+\psi_{n}\leq1/2$ and $\psi_{n}\geq C_{\psi}p\log n/n^{\kappa/4}$
with large enough $C_{\psi}>0$.\end{lem}
\begin{proof}
Denote
\[
p_{1}=\max_{s\in S_{n}}\left|\frac{e_{s}}{V_{s}}\right|\max_{s\in S_{n}}\left|\frac{V_{s}}{\hat{V}_{s}}-1\right|
\]
and
\[
p_{2}=\max_{(i,h,m)\in S_{n}}\left|\frac{\sum_{j=1}^{n}w_{h}(X_{i},X_{j})((\hat{\Sigma}_{j}^{1/2}-\Sigma_{j}^{1/2})\epsilon_{j})_{m}}{\hat{V}_{i,m,h}}\right|
\]
Then 
\[
|T^{PIA}-T^{PIA,0}|\leq p_{1}+p_{2}
\]
Let $A$ denote the event $\{\max_{j=1,...,n}\Vert\hat{\Sigma}_{j}-\Sigma_{j}\Vert_{o}<n^{-\kappa}\}$.
By assumption \ref{ass:variance estimator}, $P(A)\rightarrow1$ as
$n\rightarrow\infty.$ Thus, it is enough to show that $c_{1-\nu_{n}-\psi_{n}}^{PIA,0}\leq c_{1-\nu_{n}}^{PIA}$
and $c_{1-\nu_{n}+\psi_{n}}^{PIA,0}\geq c_{1-\nu_{n}}^{PIA}$ on $A$.

As in the proof of lemma \ref{variance}, $\max_{s\in S_{n}}|V_{s}/\hat{V}_{s}-1|\leq Cn^{-\kappa}$
on $A$. By lemma \ref{normal conv}, $E[\max_{s\in S_{n}}e_{s}/V_{s}]\leq C\sqrt{\log n}$.
So, Markov inequality gives for any $B>0$, on $A$,
\[
P(p_{1}>C\sqrt{\log n}n^{-\kappa}B|Y_{1}^{n})\leq1/B
\]
where $Y_{1}^{n}$ is a shorthand for $\{Y_{i}\}_{i=1}^{n}$. Consider
$p_{2}$. For any $j=1,...,n$ and $m=1,...,p$, 
\begin{eqnarray*}
E[((\hat{\Sigma}_{j}^{1/2}-\Sigma_{j}^{1/2})\epsilon_{j})_{m}^{2}|Y_{1}^{n}] & \leq & E[\Vert(\hat{\Sigma}_{j}^{1/2}-\Sigma_{j}^{1/2})\epsilon_{j}\Vert^{2}|Y_{1}^{n}]\\
 & \leq & E[\Vert\hat{\Sigma}_{j}^{1/2}-\Sigma_{j}^{1/2}\Vert_{o}^{2}\Vert\epsilon_{j}\Vert^{2}|Y_{1}^{n}]\\
 & \leq & p\Vert(\hat{\Sigma}_{j}^{1/2}-\Sigma_{j}^{1/2})\Vert_{o}^{2}\\
 & \leq & p^{2}\Vert\hat{\Sigma}_{j}-\Sigma_{j}\Vert_{o}
\end{eqnarray*}
where the last line follows from lemma \ref{lemma:Square root operator}.
So, conditionally on $Y_{1}^{n}$, on $A$, $\sum_{j=1}^{n}w_{h}(X_{i},X_{j})((\hat{\Sigma}_{j}^{1/2}-\Sigma_{j}^{1/2})\epsilon_{j})_{m}/V_{i,m,h}$
is mean-zero Gaussian random variable with variance bounded by $p^{2}n^{-\kappa}$
for any $(i,m,h)\in S_{n}$. In addition, on $A$, $\max_{s\in S_{n}}V_{s}/\hat{V}_{s}\leq2$
for large $n$. Thus, Markov inequality and the argument like that
used in lemma \ref{normal conv} yield
\[
P(p_{2}>C\sqrt{\log n}pn^{-\kappa/2}B|Y_{1}^{n})\leq1/B
\]
on $A$. Take $B=n^{\kappa/4}/(p\log n)$. Recall that $\psi_{n}\geq C_{\psi}p\log n/n^{\kappa/4}$.
So, $\psi_{n}>\max(4/B,\, C_{1}p^{2}(\log n)^{2}n^{-\kappa/2}B)$
for some large $C_{1}>0$ whenever $C_{1}<C_{\psi}$.

Note that $T^{PIA,0}$ is the maximum over $|S_{n}|$ standard Gaussian
random variables. In addition, for fixed $m=1,...,p$ and $h\in H_{n}$,
random variables $\{e_{i,m,h}/V_{i,m,h}:\,(i,m,h)\in S_{n}\}$ are
mutually independent, $|H_{n}|\leq C\log n$. So, lemma \ref{anticoncentration}
gives $c_{1-\nu_{n}-\psi_{n}/2}^{PIA,0}-c_{1-\nu_{n}-\psi_{n}}^{PIA,0}\geq C\psi_{n}/(p(\log n)^{3/2})$.
I will assume that $C$ in the last inequality is smaller than $C_{1}$.

Now the first part of the lemma follows from
\begin{eqnarray*}
E[g_{1-\nu_{n}-\psi_{n}}^{PIA,0}(T^{PIA})|Y_{1}^{n}] & \leq & E[g_{1-\nu_{n}-\psi_{n}}^{PIA,0}(T^{PIA,0}-p_{1}-p_{2})|Y_{1}^{n}]\\
 & \leq & E[g_{1-\nu_{n}-\psi_{n}}^{PIA,0}(T^{PIA,0}-C\sqrt{\log n}n^{-\kappa/2}B)|Y_{1}^{n}]+2/B\\
 & \leq & E[g_{1-\nu_{n}-\psi_{n}/2}^{PIA,0}(T^{PIA,0})|Y_{1}^{n}]+2/B\\
 & = & 1-\nu_{n}-\psi_{n}/2+2/B\\
 & \leq & 1-\nu_{n}
\end{eqnarray*}
 on $A$. The second part of the lemma follows from a similar argument.\end{proof}
\begin{lem}
\label{lemma: aplication of Chatterjee}$E[g_{1-\nu_{n}}^{PIA,0}(\max_{s\in S_{n}}(\varepsilon_{s}/V_{s}))]=1-\nu_{n}+o(1)$
and $E[g_{1-\nu_{n}}^{PIA,0}(-\max_{s\in S_{n}}(\varepsilon_{s}/V_{s}))]=1-\nu_{n}+o(1)$
for any sequence $\{\nu_{n}\}_{n=1}^{\infty}$ such that $\nu_{n}\in(0,1)$.\end{lem}
\begin{proof}
By lemma \ref{HS}, for any $(i,m,h)\in S_{n}$ and any $j=1,...,n$,
\[
w_{h}(X_{i},X_{j})/V_{i,m,h}\leq C/\sqrt{nh^{d}}\leq C/\sqrt{nh_{\min}^{d}}
\]
Recall the definition of $C(\cdot)$ given before theorem \ref{thm:Chatterjee}.
By assumption \ref{ass: test function}, $\beta=\beta_{n}\leq C$
for some constant $C>0$. So, $C(g_{1-\alpha}^{PIA,0})\leq C/\beta^{3}$.
In addition, $\Vert(g_{1-\alpha}^{PIA,0})^{\prime}\Vert_{\infty}\leq C/\beta$.
Given assumption \ref{ass: test function}, the result follows by
applying theorem \ref{thm:Chatterjee} with $g=g_{1-\alpha}^{PIA,0}$,
$Z_{j}=\varepsilon_{j}$, $Y_{j}=\Sigma_{j}^{1/2}\epsilon_{j}$, $a=C/\sqrt{nh_{\min}^{d}}$
and $K\leq Cn^{\phi}$ for some $\phi>0$.\end{proof}
\begin{lem}
\label{Lemma: Initial conv}$\max_{s\in S_{n}}|\varepsilon_{s}/V_{s}|=O_{p}(\sqrt{\log n})$
and $\max_{s\in S_{n}}|\varepsilon_{s}/\hat{V}_{s}|=O_{p}(\sqrt{\log n})$.\end{lem}
\begin{proof}
Combining the definition of $g_{0}$, lemma \ref{lemma: aplication of Chatterjee},
and $\beta_{n}\leq C$ for some constant $C>0$ gives
\begin{eqnarray*}
P\{\max_{s\in S_{n}}(\varepsilon_{s}/V_{s})>C\sqrt{\log n}\} & \leq & 1-E[g_{0}((\max_{s\in S_{n}}(\varepsilon_{s}/V_{s})+\beta_{n}-C\sqrt{\log n})/\beta_{n})]\\
 & = & 1-E[g_{0}((\max_{s\in S_{n}}(e_{s}/V_{s})+\beta_{n}-C\sqrt{\log n})/\beta_{n})]+o(1)\\
 & \leq & P\{\max_{s\in S_{n}}(e_{s}/V_{s})>C\sqrt{\log n}-\beta_{n}\}+o(1)\\
 & \leq & P\{\max_{s\in S_{n}}(e_{s}/V_{s})>(C/2)\sqrt{\log n}\}+o(1)
\end{eqnarray*}
By lemma \ref{normal conv}, $\max_{s\in S_{n}}(e_{s}/V_{s})=O_{p}(\sqrt{\log n})$.
So, by choosing $n$ large enough and then $C$ large enough, we can
make $P\{\max_{s\in S_{n}}(\varepsilon_{s}/V_{s})>C\sqrt{\log n}\}$
arbitrarily small uniformly in $n$. The same reasoning gives the
lower as well. We conclude that $\max_{s\in S_{n}}|\varepsilon_{s}/V_{s}|=O_{p}(\sqrt{\log n})$.
The second result follows from
\[
\max_{s\in S_{n}}|\varepsilon_{s}/\hat{V}_{s}|\leq\max_{s\in S_{n}}|\varepsilon_{s}/V_{s}|\max_{s\in S_{n}}(V_{s}/\hat{V}_{s})=O_{p}(\sqrt{\log n})
\]
since $\max_{s\in S_{n}}(V_{s}/\hat{V}_{s})=O_{p}(1)$ by lemma \ref{variance}.\end{proof}
\begin{lem}
\label{Lemma: Restriction}$P\{\max_{s\in S_{n}\backslash S_{n}^{D}}\hat{f}_{s}/\hat{V}_{s}>0\}\leq\gamma_{n}+o(1)$.\end{lem}
\begin{proof}
By lemma \ref{lemma: aplication of Chatterjee},
\[
P\{\max_{s\in S_{n}}(\varepsilon_{s}/V_{s})\leq c_{1-\gamma_{n}-\psi_{n}}^{PIA,0}+\beta_{n}\}\geq E[g_{1-\gamma_{n}-\psi_{n}}^{PIA,0}(\max_{s\in S_{n}}(\varepsilon_{s}/V_{s}))]=1-\gamma_{n}-\psi_{n}+o(1)
\]
Since for any $s\in S_{n}\backslash S_{n}^{D}$, $f_{s}/V_{s}\leq-(c_{1-\gamma_{n}-\psi_{n}}^{PIA,0}+\beta_{n})$,
\begin{eqnarray*}
P\{\max_{s\in S_{n}\backslash S_{n}^{D}}(\hat{f}_{s}/\hat{V}_{s})>0\} & = & P\{\max_{s\in S_{n}\backslash S_{n}^{D}}(\hat{f}_{s}/V_{s})>0\}\\
 & = & P\{\max_{s\in S_{n}\backslash S_{n}^{D}}(f_{s}/V_{s}+\varepsilon_{s}/V_{s})>0\}\\
 & \leq & P\{\max_{s\in S_{n}\backslash S_{n}^{D}}(-(c_{1-\gamma_{n}-\psi_{n}}^{PIA,0}+\beta_{n})+\varepsilon_{s}/V_{s})>0\}\\
 & \leq & P\{\max_{s\in S_{n}}(\varepsilon_{s}/V_{s})>c_{1-\gamma_{n}-\psi_{n}}^{PIA,0}+\beta_{n}\}\\
 & \leq & 1-(1-\gamma_{n}-\psi_{n})+o(1)\\
 & = & \gamma_{n}+\psi_{n}+o(1)
\end{eqnarray*}
Noting that $\psi_{n}=o(1)$ yields the result.\end{proof}
\begin{lem}
\label{Lemma: Inclusion}$P\{S_{n}^{D}\subset S_{n}^{RMS}\}\geq1-\gamma_{n}+o(1)$.\end{lem}
\begin{proof}
By lemma \ref{lemma: quantile approximation}, $P\{c_{1-\gamma_{n}-\psi_{n}}^{PIA,0}>c_{1-\gamma_{n}}^{PIA}\}=o(1)$.
In addition, for any $x\in(-1,1)$,
\[
2/(1+x)-1\geq2(1-x)-1\geq1-2x\geq1-2|x|
\]
So,
\begin{eqnarray*}
P\{S_{n}^{D}\subset S_{n}^{RMS}\} & = & P\{\min_{s\in S_{n}^{D}}(\hat{f}_{s}/\hat{V}_{s})>-2(c_{1-\gamma_{n}}^{PIA}+\beta_{n})\}\\
 & \geq & P\{\min_{s\in S_{n}^{D}}(\hat{f}_{s}/V_{s})\max_{s\in S_{n}^{D}}(V_{s}/\hat{V}_{s})>-2(c_{1-\gamma_{n}}^{PIA}+\beta_{n})\}\\
 & \geq & P\{\min_{s\in S_{n}^{D}}(-(c_{1-\gamma_{n}-\psi_{n}}^{PIA,0}+\beta_{n})+\varepsilon_{s}/V_{s})\max_{s\in S_{n}^{D}}(V_{s}/\hat{V}_{s})>-2(c_{1-\gamma_{n}}^{PIA}+\beta_{n})\}\\
 & = & P\{\min_{s\in S_{n}^{D}}(\varepsilon_{s}/V_{s})>c_{1-\gamma_{n}-\psi_{n}}^{PIA,0}+\beta_{n}-2(c_{1-\gamma_{n}}^{PIA}+\beta_{n})/\max_{s\in S_{n}^{D}}(V_{s}/\hat{V}_{s})\}\\
 & \geq & P\{\max_{s\in S_{n}}(-\varepsilon_{s}/V_{s})<-c_{1-\gamma_{n}-\psi_{n}}^{PIA,0}-\beta_{n}+2(c_{1-\gamma_{n}}^{PIA,0}+\beta_{n})/\max_{s\in S_{n}^{D}}(V_{s}/\hat{V}_{s})\}+o(1)\\
 & \geq & P\{\max_{s\in S_{n}}(-\varepsilon_{s}/V_{s})<(c_{1-\gamma_{n}-\psi_{n}}^{PIA,0}+\beta_{n})(1-2|\max_{s\in S_{n}^{D}}(V_{s}/\hat{V}_{s})-1|\}+o(1)
\end{eqnarray*}
Combining lemma \ref{normal conv} and Markov inequality yields
\begin{eqnarray*}
\gamma_{n}+\psi_{n} & = & 1-E[g_{1-\gamma_{n}-\psi_{n}}^{PIA,0}(\max_{s\in S_{n}}(e_{s}/V_{s}))]\\
 & \leq & P\{\max_{s\in S_{n}}(e_{s}/V_{s})>c_{1-\gamma_{n}-\psi_{n}}^{PIA,0}\}\\
 & \leq & C(\log n)^{1/2}/c_{1-\gamma_{n}-\psi_{n}}^{PIA,0}
\end{eqnarray*}
So, $c_{1-\gamma_{n}-\psi_{n}}^{PIA,0}\leq C(\log n)^{1/2}/(\gamma_{n}+\psi_{n})$.
By lemma \ref{variance}, $|\max_{s\in S_{n}^{D}}(V_{s}/\hat{V}_{s})-1|<Cn^{-\kappa}$
wpa1. So, wpa1, 
\[
(c_{1-\gamma_{n}-\psi_{n}}^{PIA,0}+\beta_{n})(1-2|\max_{s\in S_{n}^{D}}(V_{s}/\hat{V}_{s})-1|)\geq c_{1-\gamma_{n}-\psi_{n}}^{PIA,0}+\beta_{n}-C(\log n)^{1/2}n^{-\kappa}/(\gamma_{n}+\psi_{n})
\]
Take $\chi_{n}=Cp(\log n)^{2}n^{-\kappa}/(\gamma_{n}+\psi_{n})$.
Then $\chi_{n}=o(1)$ by the choice of $\psi_{n}$. By lemma \ref{anticoncentration},
\[
c_{1-\gamma_{n}-\psi_{n}}^{PIA,0}+\beta_{n}-C(\log n)^{1/2}n^{-\kappa}/(\gamma_{n}+\psi_{n})\geq c_{1-\gamma_{n}-\psi_{n}-\chi_{n}}^{PIA,0}+\beta_{n}
\]
Therefore,
\begin{eqnarray*}
P\{S_{n}^{D}\subset S_{n}^{RMS}\} & \geq & P\{\max_{s\in S_{n}}(-\varepsilon_{s}/V_{s})<c_{1-\gamma_{n}-\psi_{n}-\chi_{n}}^{PIA,0}+\beta_{n}\}+o(1)\\
 & \geq & 1-\gamma_{n}-\psi_{n}-\chi_{n}+o(1)\\
 & = & 1-\gamma_{n}+o(1)
\end{eqnarray*}
since $\psi_{n}+\chi_{n}=o(1)$.\end{proof}
\begin{lem}
\label{Lemma: Inclusion with null function}If $f=0_{p}$, then $P\{S_{n}^{RMS}=S_{n}\}\geq1-\gamma_{n}+o(1)$.\end{lem}
\begin{proof}
By lemma \ref{lemma: quantile approximation}, $P\{c_{1-\gamma_{n}-\psi_{n}}^{PIA,0}>c_{1-\gamma_{n}}^{PIA}\}=o(1)$.
By lemma \ref{variance}, $\max_{s\in S_{n}}(V_{s}/\hat{V}_{s})\leq1+n^{-\kappa}$
wpa1 as $n\rightarrow\infty$. If $f=0_{p}$, then for any $s\in S_{n}$,
$\hat{f}_{s}=\varepsilon_{s}$. So, 
\begin{eqnarray*}
P\{S_{n}^{RMS}=S_{n}\} & = & P\{\min_{s\in S_{n}}(\varepsilon_{s}/\hat{V}_{s})>-2(c_{1-\gamma_{n}}^{PIA}+\beta_{n})\}\\
 & \geq & P\{\min_{s\in S_{n}}(\varepsilon_{s}/\hat{V}_{s})>-2(c_{1-\gamma_{n}-\psi_{n}}^{PIA,0}+\beta_{n})\}+o(1)\\
 & \geq & P\{\min_{s\in S_{n}}(\varepsilon_{s}/V_{s})\max_{s\in S_{n}}(V_{s}/\hat{V}_{s})>-2(c_{1-\gamma_{n}-\psi_{n}}^{PIA,0}+\beta_{n})\}+o(1)\\
 & \geq & P\{\min_{s\in S_{n}}(\varepsilon_{s}/V_{s})(1+n^{-\kappa})>-2(c_{1-\gamma_{n}-\psi_{n}}^{PIA,0}+\beta_{n})\}+o(1)\\
 & \geq & P\{\min_{s\in S_{n}}(\varepsilon_{s}/V_{s})>-2(c_{1-\gamma_{n}-\psi_{n}}^{PIA,0}+\beta_{n})(1-n^{-\kappa})\}+o(1)\\
 & \geq & P\{\min_{s\in S_{n}}(\varepsilon_{s}/V_{s})>-(c_{1-\gamma_{n}-\psi_{n}}^{PIA,0}+\beta_{n})\}+o(1)\\
 & = & P\{\max_{s\in S_{n}}(-\varepsilon_{s}/V_{s})<(c_{1-\gamma_{n}-\psi_{n}}^{PIA,0}+\beta_{n})\}+o(1)\\
 & \geq & E[g_{1-\gamma_{n}-\psi_{n}}^{PIA,0}(\max_{s\in S_{n}}(-\varepsilon_{s}/V_{s}))]+o(1)
\end{eqnarray*}
Combining these results with lemma \ref{lemma: aplication of Chatterjee}
yields
\[
P\{S_{n}^{RMS}=S_{n}\}\geq1-\gamma_{n}-\psi_{n}+o(1)
\]
The result follows by noting that $\psi_{n}=o(1)$.\end{proof}
\begin{lem}
\label{lemma: Growth of critical value}$c_{1-\alpha}^{RMS}+\beta_{n}\leq c_{1-\alpha}^{PIA}+\beta_{n}=O_{p}(\sqrt{\log n})$.\end{lem}
\begin{proof}
Since $S_{n}^{RMS}\subseteq S_{n}$, $c_{1-\alpha}^{RMS}+\beta_{n}\leq c_{1-\alpha}^{PIA}+\beta_{n}$.
By lemma \ref{lemma: quantile approximation}, $P\{c_{1-\alpha/2}^{PIA,0}<c_{1-\alpha}^{PIA}\}=o(1)$.
By assumption \ref{ass: test function}, $\beta_{n}\leq C$ for some
$C>0$. Markov inequality and lemma \ref{normal conv}give $c_{1-\alpha/2}^{PIA,0}\leq C\sqrt{\log n}$
for $C$ large enough. Combining these results yields the statement
of the lemma. \end{proof}
\begin{lem}
\label{lemma: restricted holder}Let $\tau>1$, $L>0$, $x=(x_{1},...,x_{d})\in\mathbb{R}^{d}$,
$h=(h_{1},...,h_{d})\in\mathbb{R}^{d}$, and $f\in\mathcal{F}_{\varsigma}(\tau,L)$
for some $\varsigma=1,...,[\tau]$. If $\varsigma<[\tau]$, assume
that for any $x\in\mathbb{R}^{d}$ and all $d$-tuples of nonnegative
integers $\alpha=(\alpha_{1},...,\alpha_{d})$ satisfying $|\alpha|=\varsigma+1$,
$|D^{\alpha}f(x)|\leq C$ for some constant $C>0$. Then $\partial f(x_{1},...,x_{d})/\partial x_{m}\geq0$
for all $m=1,...,d$ implies that for any $y=(y_{1},...,y_{d})\in\mathbb{R}^{d}$
satisfying $0\leq y\leq h$, 
\[
f(x+y)-f(x)\geq-\frac{\max(L^{\tau-[\tau]},C)}{(\tau-\varsigma+1)...(\tau-\varsigma+\varsigma)}\Vert h\Vert^{\zeta}
\]
for $\zeta=\min(\varsigma+1,\tau)$.\end{lem}
\begin{proof}
For any $y=(y_{1},...,y_{d})\in\mathbb{R}^{d}$ satisfying $0\leq y\leq h$,
choose a direction $l=(l_{1},...,l_{d})\in\mathbb{R}^{d}$ by setting
$l_{m}=y_{m}/\sqrt{\sum_{j=1}^{d}y_{j}^{2}}$ for all $m=1,...,d$.
Let $f^{(k,l)}(x)$ denotes $k$-th derivative of $f$ in direction
$l$ evaluated at point $x$. Then $f^{(1,l)}(x)\geq0$. If $f^{(1,l)}(x+ty)\geq0$
for all $t\in(0,1)$, then the result is obvious. If $f^{(1,l)}(x+t_{0}y)=0$
for some $t_{0}\in(0,1)$, then $f^{(k,l)}(x+t_{0}y)=0$ for all $k=1,...,\varsigma$.
If $\varsigma=[\tau]$, then by Holder smoothness, $f^{(\varsigma,l)}(x+ty)\geq-(L(t-t_{0})\Vert y\Vert)^{\tau-\varsigma}$.
Integrating it $\varsigma$ times gives 
\begin{equation}
f(x+y)-f(x)\geq-\frac{L^{\tau-[\tau]}}{(\gamma-\varsigma+1)...(\gamma-\varsigma+K)}\Vert y\Vert^{\zeta}\label{eq: holder derivation}
\end{equation}
since $\zeta=\tau$ in this case. If $\varsigma<[\tau]$, then $f^{(\varsigma,l)}(x+ty)\geq-C(t-t_{0})\Vert y\Vert$.
Integrating it $\varsigma$ times gives the inequality similar to
(\ref{eq: holder derivation}) with $\varsigma+1$ instead of $\zeta$
and $C$ instead of $L^{\tau-\varsigma}$. The result follows by noting
that $\Vert y\Vert\leq\Vert h\Vert$.
\end{proof}

\subsection{\label{sub:Proofs-of-Theorems}Proofs of Theorems}
\begin{proof}[Proof of Theorem 1]
Under the null hypothesis, for any $s\in S_{n}$, $f_{s}\leq0$ since
the kernel $K$ is positive by assumption \ref{ass:The-kernel}. By
lemma \ref{lemma: quantile approximation}, $P\{c_{1-\alpha-\psi_{n}}^{PIA,0}>c_{1-\alpha}^{PIA}\}=o(1)$.
By lemma \ref{variance}, $\max_{s\in S_{n}}(V_{s}/\hat{V}_{s})\leq1+n^{-\kappa}$
wpa1 as $n\rightarrow\infty$. So,
\begin{eqnarray*}
E[g_{1-\alpha}^{PIA}(\hat{T})] & = & E[g_{1-\alpha}^{PIA}(\max_{s\in S_{n}}(\hat{f}_{s}/\hat{V}_{s}))]\\
 & \geq & E[g_{1-\alpha}^{PIA}(\max_{s\in S_{n}}(\varepsilon_{s}/\hat{V}_{s}))]\\
 & \geq & E[g_{1-\alpha-\psi_{n}}^{PIA,0}(\max_{s\in S_{n}}(\varepsilon_{s}/\hat{V}_{s}))]+o(1)\\
 & \geq & E[g_{1-\alpha-\psi_{n}}^{PIA,0}(\max_{s\in S_{n}}(\varepsilon_{s}/V_{s})\max_{s\in S_{n}}(V_{s}/\hat{V}_{s}))]+o(1)\\
 & \geq & E[g_{1-\alpha-\psi_{n}}^{PIA,0}(\max_{s\in S_{n}}(\varepsilon_{s}/V_{s})(1+n^{-\kappa}))]+o(1)\\
 & \geq & E[g_{0}((\max_{s\in S_{n}}(\varepsilon_{s}/V_{s})(1+n^{-\kappa})-c_{1-\alpha-\psi_{n}}^{PIA,0})/\beta_{n})]+o(1)
\end{eqnarray*}
Denote $\delta_{n}=(\log n/n^{\kappa})^{1/2}$. Two different cases
will be considered depending on whether $\beta_{n}>\delta_{n}$ or
$\beta_{n}\leq\delta_{n}$. Divide the sequence $\{n\}_{n=1}^{\infty}$
into two subsequences, $\{n_{k}^{1}\}_{k=1}^{\infty}$ and $\{n_{k}^{2}\}_{k=1}^{\infty}$,
so that $\beta_{n_{k}^{1}}>\delta_{n_{k}^{1}}$ and $\beta_{n_{k}^{2}}\leq\delta_{n_{k}^{2}}$
for all $k\in\mathbb{N}$. First, consider the subsequence $\{n_{k}^{1}\}_{k=1}^{\infty}$.
For simplicity of notation, I will drop indices writing $n$ instead
of $n_{k}^{1}$. By lemma \ref{Lemma: Initial conv}, $\max_{s\in S_{n}}|\varepsilon_{s}/V_{s}|=O_{p}(\sqrt{\log n})$.
So, $\max_{s\in S_{n}}|\varepsilon_{s}/V_{s}|/(n^{\kappa}\beta_{n})<n^{-\kappa/4}$
wpa1 as $n\rightarrow\infty$. Since $g_{0}$ has bounded first derivative,
\begin{multline*}
E[g_{0}((\max_{s\in S_{n}}(\varepsilon_{s}/V_{s})(1+n^{-\kappa})-c_{1-\alpha-\psi_{n}}^{PIA,0})/\beta_{n})]\\
=E[g_{0}((\max_{s\in S_{n}}(\varepsilon_{s}/V_{s})-c_{1-\alpha-\psi_{n}}^{PIA,0})/\beta_{n})]+o(1)
\end{multline*}
The last expression equals $E[g_{1-\alpha-\psi_{n}}^{PIA,0}(\max_{s\in S_{n}}(\varepsilon_{s}/V_{s}))]+o(1$).
Combining these results and lemma \ref{lemma: aplication of Chatterjee}
yields
\[
E[g_{1-\alpha}^{PIA}(\hat{T})]\geq1-\alpha-\psi_{n}+o(1)=1-\alpha+o(1)
\]

Next, consider the subsequence $\{n_{k}^{2}\}_{k=1}^{\infty}$. Again,
I will write $n$ instead of $n_{k}^{2}$. Take $\chi_{n}=Cp(\log n)^{2}n^{-\kappa/2}$
with large enough $C$. Note that $\chi_{n}=o(1)$. As in lemma \ref{Lemma: Inclusion},
\[
c_{1-\alpha-\psi_{n}}^{PIA,0}(1-n^{-\kappa})-\beta_{n}\geq c_{1-\alpha-\psi_{n}-\chi_{n}}^{PIA,0}
\]
Continuing the chain of inequalities from above gives 
\begin{eqnarray*}
E[g_{1-\alpha}^{PIA}(\hat{T})] & \geq & P\{\max_{s\in S_{n}}(\varepsilon_{s}/V_{s})(1+n^{-\kappa})\leq c_{1-\alpha-\psi_{n}}^{PIA,0}\}+o(1)\\
 & \geq & P\{\max_{s\in S_{n}}(\varepsilon_{s}/V_{s})\leq c_{1-\alpha-\psi_{n}}^{PIA,0}(1-n^{-\kappa})\}+o(1)\\
 & \geq & P\{\max_{s\in S_{n}}(\varepsilon_{s}/V_{s})-\beta_{n}\leq c_{1-\alpha-\psi_{n}-\chi_{n}}^{PIA,0}\}+o(1)\\
 & \geq & E[g_{1-\alpha-\psi_{n}-\chi_{n}}^{PIA,0}(\max_{s\in S_{n}}(\varepsilon_{s}/V_{s}))]
\end{eqnarray*}
An application of lemma \ref{lemma: aplication of Chatterjee} yields
\[
E[g_{1-\alpha}^{PIA}(\hat{T})]\geq1-\alpha-\psi_{n}-\chi_{n}+o(1)=1-\alpha+o(1)
\]

Now consider the RMS test function. By lemma \ref{Lemma: Inclusion},
$P\{c_{1-\alpha+2\gamma_{n}}^{D}>c_{1-\alpha+2\gamma_{n}}^{RMS}\}\leq\gamma_{n}+o(1)$.
By lemma \ref{Lemma: Restriction}, $P\{\max_{s\in S_{n}\backslash S_{n}^{D}}\hat{f}_{s}/\hat{V}_{s}>0\}\leq\gamma_{n}+o(1)$.
So,
\begin{eqnarray*}
E[g_{1-\alpha+2\gamma_{n}}^{RMS}(\hat{T})] & = & E[g_{1-\alpha+2\gamma_{n}}^{RMS}(\max_{s\in S_{n}}(\hat{f}_{s}/\hat{V}_{s}))]\\
 & \geq & E[g_{1-\alpha+2\gamma_{n}}^{D}(\max_{s\in S_{n}}(\hat{f}_{s}/\hat{V}_{s}))]-\gamma_{n}+o(1)\\
 & \geq & E[g_{1-\alpha+2\gamma_{n}}^{D}(\max_{s\in S_{n}^{D}}(\hat{f}_{s}/\hat{V}_{s}))]-2\gamma_{n}+o(1)
\end{eqnarray*}
Since $S_{n}^{D}$ is nonstochastic, from this point, the argument
similar to that used in the proof for the plug-in test function with
$S_{n}^{D}$ instead of $S_{n}$ yields the result for the RMS critical
values.

Next assume that $f=0_{p}$. By lemma \ref{lemma: quantile approximation},
$P\{c_{1-\alpha+\psi_{n}}^{PIA,0}<c_{1-\alpha}^{PIA}\}=o(1)$. By
lemma \ref{variance}, $\min_{s\in S_{n}}(V_{s}/\hat{V}_{s})\geq1-n^{-\kappa}$
wpa1 as $n\rightarrow\infty$. So,
\begin{eqnarray*}
E[g_{1-\alpha}^{PIA}(\hat{T})] & = & E[g_{1-\alpha}^{PIA}(\max_{s\in S_{n}}(\hat{f}_{s}/\hat{V}_{s}))]\\
 & = & E[g_{1-\alpha}^{PIA}(\max_{s\in S_{n}}(\varepsilon_{s}/\hat{V}_{s}))]\\
 & \leq & E[g_{1-\alpha+\psi_{n}}^{PIA,0}(\max_{s\in S_{n}}(\varepsilon_{s}/\hat{V}_{s}))]+o(1)\\
 & \leq & E[g_{1-\alpha+\psi_{n}}^{PIA,0}(\max_{s\in S_{n}}(\varepsilon_{s}/V_{s})\min_{s\in S_{n}}(V_{s}/\hat{V}_{s}))]+o(1)\\
 & \leq & E[g_{1-\alpha+\psi_{n}}^{PIA,0}(\max_{s\in S_{n}}(\varepsilon_{s}/V_{s})(1-n^{-\kappa}))]+o(1)\\
 & = & E[g_{0}((\max_{s\in S_{n}}(\varepsilon_{s}/V_{s})(1-n^{-\kappa})-c_{1-\alpha+\psi_{n}}^{PIA,0})/\beta_{n})]+o(1)
\end{eqnarray*}
For the subsequence $\{n_{k}^{1}\}_{k=1}^{\infty}$, writing $n$
instead of $n_{k}^{1}$,
\begin{multline*}
E[g_{0}((\max_{s\in S_{n}}(\varepsilon_{s}/V_{s})(1-n^{-\kappa})-c_{1-\alpha+\psi_{n}}^{PIA,0}))/\beta_{n}]\\
=E[g_{0}((\max_{s\in S_{n}}(\varepsilon_{s}/V_{s})-c_{1-\alpha+\psi_{n}}^{PIA,0})/\beta_{n})]+o(1)
\end{multline*}
So, the result that $E[g_{1-\alpha}^{PIA}(\hat{T})]\leq1-\alpha+o(1)$
follows by applying \ref{lemma: aplication of Chatterjee}. For the
subsequence $\{n_{k}^{2}\}_{k=1}^{\infty}$, with the same choice
of $\chi_{n}$,
\[
(c_{1-\alpha+\psi_{n}}^{PIA,0}+\beta_{n})(1+2n^{-\kappa})\leq c_{1-\alpha+\psi_{n}+\chi_{n}}^{PIA,0}
\]
where I again write $n$ instead of $n_{k}^{2}$. In addition, for
any $x\in(0,1/2)$,
\[
1/(1-x)<1+2x
\]
So,
\begin{eqnarray*}
E[g_{1-\alpha}^{PIA}(\hat{T})] & \leq & P\{\max_{s\in S_{n}}(\varepsilon_{s}/V_{s})(1-n^{-\kappa})\leq c_{1-\alpha+\psi_{n}}^{PIA,0}+\beta_{n}\}+o(1)\\
 & \leq & P\{\max_{s\in S_{n}}(\varepsilon_{s}/V_{s})\leq(c_{1-\alpha+\psi_{n}}^{PIA,0}+\beta_{n})(1+2n^{-\kappa})\}+o(1)\\
 & \leq & P\{\max_{s\in S_{n}}(\varepsilon_{s}/V_{s})\leq c_{1-\alpha+\psi_{n}+\chi_{n}}^{PIA,0}\}+o(1)\\
 & \leq & E[g_{1-\alpha+\psi_{n}+\chi_{n}}^{PIA,0}(\max_{s\in S_{n}}(\varepsilon_{s}/V_{s}))]
\end{eqnarray*}
Again, the result that $E[g_{1-\alpha}^{PIA}(\hat{T})]\leq1-\alpha+o(1)$
follows by applying lemma \ref{lemma: aplication of Chatterjee}.

For the RMS test function, note that by lemma \ref{Lemma: Inclusion with null function},
$P\{S_{n}^{RMS}=S_{n}\}\geq1-\gamma_{n}+o(1)$ whenever $f=0_{p}$.
If $\gamma_{n}=o(1)$, then 
\[
E[g_{1-\alpha}^{RMS}(\hat{T})]=E[g_{1-\alpha+2\gamma_{n}}^{PIA}(\hat{T})]+o(1)\leq1-\alpha+o(1)
\]

\end{proof}
Proof of Corollary 1:
\begin{proof}[Proof of Corollary 1]
If $(\log n)^{19}/(h_{\min}^{3d}n)\rightarrow0$, then one can set
$\varrho_{n}$ so that $\varrho_{n}(\log n)^{3/2}\rightarrow0$ and
$(\log n)^{4}/(\varrho_{n}^{10}h_{\min}^{3d}n)\rightarrow0$. Then
$\varrho_{n}$ satisfies assumption \ref{ass: test function}. So,
the result of theorem \ref{thm: size} holds for $\varrho_{n}$ instead
of $\beta_{n}$. Let $c_{x}^{PIA,0,\varrho}$ denote the value of
$c_{x}^{PIA,0}$ evaluated with $\varrho_{n}$ instead of $\beta_{n}$
for all $x\in(0,1)$. By lemma \ref{lemma: quantile approximation},
$P\{c_{1-\alpha-\psi_{n}}^{PIA,0}>c_{1-\alpha}^{PIA}\}=o(1)$. By
lemma \ref{anticoncentration}, $c_{1-\alpha-\psi_{n}-C\varrho_{n}(\log n)^{3/2}}^{PIA,0,\varrho}+\varrho_{n}\leq c_{1-\alpha-\psi_{n}}^{PIA,0}$
for $C$ large enough. So,
\begin{eqnarray*}
P\{\hat{T}\leq c_{1-\alpha}^{PIA}\} & \geq & E[g_{0}((\hat{T}+\varrho_{n}-c_{1-\alpha}^{PIA})/\varrho_{n})]\\
 & \geq & E[g_{0}((\hat{T}+\varrho_{n}-c_{1-\alpha-\psi_{n}}^{PIA,0})/\varrho_{n})]+o(1)\\
 & \geq & E[g_{0}((\hat{T}-c_{1-\alpha-\psi_{n}-C\varrho_{n}(\log n)^{3/2}}^{PIA,0,\varrho})/\varrho_{n})]+o(1)
\end{eqnarray*}
From this point, the argument like that used in the proof of theorem
\ref{thm: size} with $\varrho_{n}$ instead of $\beta_{n}$ leads
to $P\{\hat{T}\leq c_{1-\alpha}^{PIA}\}\geq1-\alpha+o(1)$. All other
statements of theorem \ref{thm: size} follow from similar arguments.
\end{proof}
Proof of Theorem 2:
\begin{proof}[Proof of Theorem 2]
For any $w\in\mathcal{G}_{\rho}$, there exist $i(w)\in\mathbb{N}$
and $m(w)=1,...,p$ such that $f_{m(w)}^{w}(X_{i(w)})\geq\rho$. For
simplicity of notation, I will drop index $w$. By assumption \ref{ass:Regression-function},
there exists a ball $B_{\delta}(X_{i})$ with center at $X_{i}$ and
radius $\delta$ such that $f_{m}(X_{j})\geq\rho/2$ for all $X_{j}\in B_{\delta}(X_{i})$.
Note that $\delta$ can be chosen independently of $w$. So, for some
$N\in\mathbb{N}$ and any $n\geq N$, there exists a triple $s_{n}=(i_{n},m,h_{n})\in S_{n}$
with $h_{n}$ bounded away from zero such that $f_{m}(X_{j})\geq\rho/2$
for all $X_{j}\in B_{h_{n}}(X_{i_{n}})$. Hence, $f_{s_{n}}\geq\rho/2$.
Lemma \ref{HS} gives $V_{s_{n}}\leq n^{-\phi}$ for some $\phi>0$,
so $f_{s_{n}}/V_{s_{n}}>Cn^{\phi}$. By lemma \ref{variance}, $|\hat{V}_{s_{n}}/V_{s_{n}}-1|=o_{p}(1)$.
So, for any $\tilde{C}<C$, $P\{f_{s_{n}}/\hat{V}_{s_{n}}>\tilde{C}n^{\phi}\}\rightarrow1$.
Thus,
\begin{eqnarray*}
E[g_{1-\alpha}^{P}(\hat{T})] & \leq & P\{\hat{T}\leq c_{1-\alpha}^{P}+\beta_{n}\}\\
 & \leq & P\{f_{s_{n}}/\hat{V}_{s_{n}}\leq c_{1-\alpha}^{P}+\beta_{n}+\max_{s\in S_{n}}|\varepsilon_{s}/\hat{V}_{s}|\}\\
 & \leq & P\{c_{1-\alpha}^{P}+\beta_{n}+\max_{s\in S_{n}}|\varepsilon_{s}/\hat{V}_{s}|>\tilde{C}n^{\phi}\}+o(1)
\end{eqnarray*}
The result follows by noting that from lemmas \ref{Lemma: Initial conv}
and \ref{lemma: Growth of critical value}, $c_{1-\alpha}^{P}+\beta_{n}+\max_{s\in S_{n}}|\varepsilon_{s}/\hat{V}_{s}|=O_{p}(\sqrt{\log n})$.
\end{proof}
Proof of Theorem 3:
\begin{proof}[Proof of Theorem 3]
As in the proof of theorem 2, since $\rho(w,H_{0})>0$, there exists
$i\in\mathbb{N}$ such that $f_{m}^{w}(X_{i})\geq\rho$ for some $m=1,...,p$
and $\rho>0$. In addition, by assumption \ref{ass:Regression-function},
there exists a ball $B_{\delta}(X_{i})$ such that $f_{m}^{w}(X_{j})\geq\rho/2$
for all $X_{j}\in B_{\delta}(X_{i})$. So, for some $N\in\mathbb{N}$
and any $n\geq N$, there exists a triple $s_{n}=(i_{n},m,h)\in S_{n}$
such that $f_{m}^{w}(X_{j})\geq\rho/2$ for all $X_{j}\in B_{h}(X_{i_{n}})$.
Hence, $f_{s_{n}}^{n}\geq a_{n}\rho/2$. Note that in contrast with
theorem 2, now we choose fixed bandwidth value $h$. By lemma \ref{HS},
$V_{s_{n}}\leq C/\sqrt{n}$. Then lemma \ref{variance} gives $P\{f_{s_{n}}^{n}/\hat{V}_{s_{n}}>\tilde{C}a_{n}/\sqrt{n}\}\rightarrow1$
for some $\tilde{C}>0$. The same argument as in the proof of theorem
2 yields
\[
E[g_{1-\alpha}^{P}(\hat{T})]\leq P\{c_{1-\alpha}^{P}+\beta_{n}+\max_{s\in S_{n}}|\varepsilon_{s}/\hat{V}_{s}|>\tilde{C}a_{n}\sqrt{n}\}+o(1)
\]
Combining $c_{1-\alpha}^{P}+\beta_{n}+\max_{s\in S_{n}}|\varepsilon_{s}/\hat{V}_{s}|=O_{p}(\sqrt{\log n})$
and $a_{n}\sqrt{n/\log n}\rightarrow\infty$ gives the result.
\end{proof}
Proof of Theorem 4:
\begin{proof}[Proof of Theorem 4]
First, consider $\tau\leq1$ case. In this case, $\zeta=\tau$. Since
$d\geq1$, we are in the situation $\zeta\leq d$. For any $w\in\mathcal{G}_{\vartheta}$,
there exist $i(w)\in\mathbb{N}_{\vartheta}$ and $m(w)=1,...,p$ such
that $f_{m(w)}^{w}(X_{i(w)})\geq(C/2)h_{\min}^{\zeta}$. By assumptions
\ref{ass:Design-points} and \ref{ass:Choice function}, there exists
$j(w)=1,...,n$ such that $\Vert X_{i(w)}-X_{j(w)}\Vert\leq3h_{\min}$
and $s_{n}(w)=(j(w),m(w),h_{\min})\in S_{n}$. By assumption \ref{ass:Regression-function},
$f_{m(w)}^{w}(X_{l})\geq\tilde{C}h_{\min}^{\zeta}$ for all $l=1,...,n$
such that $X_{l}\in B_{h_{\min}}(X_{j(w)})$ for some constant $\tilde{C}$.
So, $f_{s_{n}(w)}^{w}\geq\tilde{C}h_{\min}^{\zeta}$. By assumption
\ref{ass: test function}, $nh_{\min}^{3d}/\log n\rightarrow\infty$
as $n\rightarrow\infty$. By lemma \ref{HS}, $V_{s_{n}(w)}\leq C/\sqrt{nh_{\min}^{d}}$.
So,
\[
f_{s_{n}(w)}^{w}/(V_{s_{n}(w)}\sqrt{\log n})\geq(\tilde{C}/C)\sqrt{nh_{\min}^{2\zeta+d}/\log n}\geq(\tilde{C}/C)\sqrt{nh_{\min}^{3d}/\log n}\rightarrow\infty
\]
uniformly in $w\in\mathcal{G}_{\vartheta}$. The result follows from
the same argument as in the proof of theorem 2.

Consider $\tau>1$ case. Suppose $\zeta\leq d$. For any $w\in\mathcal{G}_{\vartheta}$,
there exist $i(w)\in\mathbb{N}_{\vartheta}$ and $m(w)=1,...,p$ such
that $f_{m(w)}^{w}(X_{i(w)})\geq(C/2)h_{\min}^{\zeta}$. For $m=1,...,d$,
set $e_{m}=4h_{\min}$ if $\partial f_{m(w)}^{w}(X_{i(w)})/\partial x_{m}\geq0$
and $-4h_{\min}$ otherwise. Consider the cube $\mathcal{C}$ whose
edges are parallel to axes and that contains vertices $(X_{i(w),1},...,X_{i(w),d})$
and $(X_{i(w),1}+2e_{1},...,X_{i(w),d}+2e_{d})$. By lemma \ref{lemma: restricted holder},
for all $x\in\mathcal{C}$, $f_{m(w)}^{w}(x)\geq\tilde{C}h_{\min}^{\zeta}$
for some constant $\tilde{C}$. By the definition of $\mathbb{N}_{\vartheta}$
and assumption \ref{ass:Design-points}, there exists $l(w)=1,...,n$
such that $X_{l(\omega)}\in B_{h_{\min}}(X_{i(w),1}+e_{1},...,X_{i(w),d}+e_{d})$.
By assumption \ref{ass:Choice function}, there exists $j(w)=1,...,n$
such that $X_{j(w)}\in B_{3h_{\min}}(X_{i(w),1}+e_{1},...,X_{i(w),d}+e_{d})$
and $s_{n}(w)=(j(w),m(w),h_{\min})\in S_{n}$. So, $f_{m(w)}^{w}(X_{l})\geq\tilde{C}h_{\min}^{\zeta}$
for all $l=1,...,n$ such that $X_{l}\in B_{h_{\min}}(X_{j(w)})$.
The rest of the proof follows from the same argument as in the case
$\tau\leq1$.

Suppose $\zeta>d$. The only difference between this case and the
previous one is that now optimal testing bandwidth value is greater
than $h_{\min}$. Let $h_{o}$ be the largest bandwidth value in the
set $S_{n}$ which is smaller than $(\log n/n)^{1/(2\zeta+d)}$. For
any $w\in\mathcal{G}_{\vartheta}$, the same construction as above
gives $s_{n}(w)=(j(w),m(w),h_{o})\in S_{n}$ such that $f_{m(w)}^{w}(X_{l})\geq\rho_{\vartheta}(w,H_{0})-\tilde{C}h_{o}^{\zeta}$
for all $l=1,...,n$ such that $X_{l}\in B_{h_{o}}(X_{j(w)})$. Since
$\rho_{\vartheta}(w,H_{0})\geq b_{n}(\log n/n)^{\zeta/(2\zeta+d)}$
for some sequence of real numbers $\{b_{n}\}_{n=1}^{\infty}$ such
that $b_{n}\rightarrow\infty$ as $n\rightarrow\infty$, $f_{s_{n}(w)}^{w}\geq(b_{n}-\tilde{C})(\log n/n)^{\zeta/(2\zeta+d)}$.
By lemma \ref{HS}, $V_{s_{n}(w)}\leq C/\sqrt{nh_{o}^{d}}$. Then
\[
f_{s_{n}(w)}^{w}/(V_{s_{n}(w)}\sqrt{\log n})\geq(b_{n}-\tilde{C})/(2C)\rightarrow\infty
\]
The result follows as above.
\end{proof}
Proof of Theorem 5:
\begin{proof}[Proof of Theorem 5]
Define $v:\,\mathbb{R}\times\mathbb{R}_{+}\rightarrow\mathbb{R}_{+}$
as follows. Set $v(x,h)=0$ if $x<0$ or $x>2$ for all $h\in\mathbb{R}_{+}$. 

First, define functions $b_{1},...,b_{K}$ on $(0,1]$ for some $K$
to be chosen below by the following induction. Set $b_{1}(x)=+1$
for $x\in(0,1/2]$ and $-1$ for $x\in(1/2,1]$. Given $b_{1},...,b_{k-1}$,
for $i=1,3,...,2^{k}-1$ and $x\in((i-1)2^{-k},i2^{-k}]$, set $b_{k}(x)=+1$
if $b_{k-1}(y)=+1$ for $y\in((i-1)2^{-k},(i+1)2^{-k}]$ and $-1$
otherwise. For $i=2,4,...,2^{k}$ and $x\in((i-1)2^{-k},i2^{-k}]$,
set $b_{k}(x)=-1$ if $b_{k-1}(y)=+1$ for $y\in((i-2)2^{-k},i2^{-k}]$
and $+1$ otherwise. By induction, define $b_{1},...,b_{K}$ where
$K$ is the largest integer strictly smaller than $\tau$, i.e. $K=[\tau]$.

Now let us define $\nu:\,\mathbb{R}\times\mathbb{R}_{+}\rightarrow\mathbb{R}_{+}$.
Set $v(x,h)=0$ if $x<0$ or $x>2$ for all $h\in\mathbb{R}_{+}$.
For $x\in[0,2]$, $\nu$ will be defined through its derivatives.
Set $\partial^{k}v(0,h)/\partial x^{k}=0$ for all $k=0,...,K$. For
$i=1,...,2^{K}$, once function $\partial^{K}v(x,h)/\partial x^{K}$
is defined for $x\in[0,(i-1)2^{-K}]$, set
\[
\partial^{K}v(x,h)/\partial x^{K}=\partial^{K}v((i-1)2^{-K},h)/\partial x^{K}+b_{K}(x)h^{K}L(x-(i-1)2^{-K})^{\tau-K}
\]
for $x\in((i-1)2^{-K},i2^{-K}]$. These conditions define function
$v(x,h)$ for $x\in[0,1]$ and $h\in\mathbb{R}_{+}$. For $x\in(1,2]$
and $h\in\mathbb{R}_{+}$, set $v(x,h)=v(2-x,h)$ so that $v$ is
symmetric in $x$ around $x=1$. It is easy to see that for fixed
$h\in\mathbb{R}_{+}$, $v(\cdot/h,h)\in\mathcal{F}_{[\tau]}(\tau,L)$
and $\sup_{x\in\mathbb{R}}v(x/h,h)\in(C_{1}h^{\tau},C_{2}h^{\tau})$
for some positive constants $C_{1}$ and $C_{2}$ independent of $h$.

Let $q:\,\mathbb{R}^{d}\times\mathbb{R}_{+}\rightarrow\mathbb{R}_{+}$
be given by $q(x,h)=v(\Vert x\Vert/h+1,h)$ for all $(x,h)\in\mathbb{R}^{d}\times\mathbb{R}_{+}$.
Note that for fixed $h\in\mathbb{R}_{+}$, $q(\cdot,h)\in\mathcal{F}_{[\tau]}(\tau,L)$,
$q(x,h)=0$ if $\Vert x\Vert>h$, and $q(0_{d},h)=\sup_{x\in\mathbb{R}^{d}}q(x,h)\in(C_{1}h^{\tau},C_{2}h^{\tau})$.

Since $r_{n}(n/\log n)^{\tau/(2\tau+d)}\rightarrow0$, there exists
a sequence of positive numbers $\{\psi_{n}\}_{n=1}^{\infty}$ such
that $r_{n}=\psi_{n}^{\tau}(\log n/n)^{\tau/(2\tau+d)}$ and $\psi_{n}\rightarrow0$.
Set $h_{n}=\psi_{n}(\log n/n)^{1/(2\tau+d)}$. By the assumption on
packing numbers $N(h,S_{\vartheta})$, there exists a set $\{j(l)\in\mathbb{N}_{\vartheta}:\, l=1,...,N_{n}\}$
such that $\Vert X_{j(l_{1})}-X_{j(l_{2})}\Vert>2h_{n}$ for $l_{1},l_{2}=1,...,N_{n}$
if $l_{1}\neq l_{2}$ and $N_{n}>Ch_{n}^{-d}$ for some constant $C$.
For $l=1,...,N_{n}$, define function $f^{l}:\,\mathbb{R}^{d}\rightarrow\mathbb{R}^{p}$
given by $f_{1}^{l}(x)=q(x-X_{j(l)},h_{n})$ and $f_{m}^{l}(x)=0$
for all $m=2,...,p$ for all $x\in\mathbb{R}^{d}$. Note that functions
$\{f^{l}\}_{l=1}^{N_{n}}$ have disjoint supports. Moreover, for every
$l=1,...,N_{n}$ and $m=1,...,p$, $f_{m}^{l}\in\mathcal{F}_{[\tau]}(\tau,L)$.
Let $\{\varepsilon_{i}\}_{i=1}^{n}$ be a sequence of independent
standard Gaussian random vectors $N(0,I_{p})$. For $l=1,...,N_{n}$,
define an alternative, $w_{l}$, with the regression function $f^{l}$
and disturbances $\{\varepsilon_{i}\}_{i=1}^{n}$. Note that $\rho_{\vartheta}(w_{l},H_{0})\geq Cr_{n}$
for all $l=1,...,N_{n}$ for some constant $C$. In addition, let
$w_{0}$ denote the alternative with zero regression function and
disturbances $\{\varepsilon_{i}\}_{i=1}^{n}$.

As in the proof of lemma 6.2 in \citet{Dumbgen2001}, for any sequence
$\phi_{n}=\phi_{n}(Y_{1},...,Y_{n})$ of tests with $\sup_{w\in\mathcal{G}_{0}}E_{w}[\phi_{n}]\leq\alpha$,
\begin{eqnarray*}
\inf_{w\in\mathcal{G},\rho_{\vartheta}(w,H_{0})\geq Cr_{n}}E_{w}[\phi_{n}]-\alpha & \leq & \min_{l=1,...,N_{n}}E_{w_{l}}[\phi_{n}]-E_{w_{0}}[\phi_{n}]\\
 & \leq & \sum_{i=1}^{N_{n}}E_{w_{l}}[\phi_{n}]/N_{n}-E_{w_{0}}[\phi_{n}]\\
 & \leq & E_{w_{0}}[(\sum_{i=1}^{N_{n}}(dP_{w_{l}}/dP_{w_{0}})/N_{n}-1)\phi_{n}]\\
 & \leq & E_{w_{0}}[|\sum_{i=1}^{N_{n}}dP_{w_{l}}/dP_{w_{0}}/N_{n}-1|]
\end{eqnarray*}
where $dP_{w_{l}}/dP_{w_{0}}$ denotes a Radon-Nykodim derivative.
For $l=1,...,N_{n}$, denote $\omega_{l}=(\sum_{i=1}^{n}(f_{1}^{l}(X_{i}))^{2})^{1/2}$
and $\xi_{l}=\sum_{i=1}^{n}f_{1}^{l}(X_{i})Y_{i,1}/\omega_{l}$. Then
\[
dP_{w_{l}}/dP_{w_{0}}=\exp(w_{l}\xi_{l}-\omega_{l}^{2}/2)
\]
Note that $\omega_{l}\leq Cn^{1/2}h_{n}^{\tau+d/2}$. In addition,
under the model $w_{0}$, $\xi_{l}$ are independent standard Gaussian
random variables. So, an application of lemma \ref{property of Gauss rv}
gives
\[
E_{w_{0}}[|\sum_{i=1}^{N_{n}}dP_{w_{l}}/dP_{w_{0}}/N_{n}-1|]\rightarrow0
\]
if $Cn^{1/2}h_{n}^{\tau+d/2}<\tilde{C}(\log N_{n})^{1/2}$ for some
constant $\tilde{C}\in(0,1)$ for all large enough $n$. The result
follows by noting that $n^{1/2}h_{n}^{\tau+d/2}=o(\sqrt{\log n})$
and $\log N_{n}\geq C\log n$ for some constant $C$.
\end{proof}
Proof of Corollary 2:
\begin{proof}[Proof of Corollary 2]
Replace $p$ by $K_{n}$ both in $\psi_{n}$ and $\chi_{n}$ in all
preliminary results and theorem \ref{thm: size}. Then all preliminary
results except lemma \ref{lemma: aplication of Chatterjee} hold for
the test with $K_{n}\rightarrow\infty$. Lemma \ref{lemma: aplication of Chatterjee}
holds with conditions (iii) and (iv) in the corollary replacing assumption
\ref{ass: test function}. So, the first result follows from the same
argument as in theorem \ref{thm: size}. For any $w\in\mathcal{G}_{\rho}$,
there exists some $m(w)\in\mathbb{N}$ such that $\sup_{i\in\mathbb{N}}[f_{m(w)}^{w}(X_{i})]_{+}>0$.
Once $m(w)$ is included in the test statistic, the second result
follows as in the proof of theorem \ref{Theorem: fixed alternatives}.
\end{proof}
\bibliographystyle{elsarticle-harv}
\bibliography{C:/Denis/Research/Bibliography}

\end{document}